\documentclass[twocolumn]{IEEEtran}

\normalsize

\ifCLASSINFOpdf
\else
\fi

\usepackage{stfloats}

\usepackage{mathtools}
\usepackage{amsmath}
\usepackage{amssymb}
\usepackage{amsmath}
\usepackage{cite}
\usepackage{siunitx}
\usepackage{mathtools}
\usepackage{amsmath}
\usepackage{amssymb}
\usepackage{amsmath}
\usepackage{cite}
\usepackage{siunitx}

\usepackage{amsthm}
\usepackage{graphicx}
\usepackage{graphicx}
\usepackage[utf8]{inputenc}

\usepackage{amsthm}

\usepackage{epstopdf}
\usepackage{mdwmath}
\usepackage{blindtext}
\usepackage{eqparbox}
\usepackage{fixltx2e}

\usepackage{graphicx}
\usepackage{graphicx}

\usepackage{mdwmath}

\usepackage{eqparbox}
\usepackage{fixltx2e}
\DeclareMathOperator*{\argmin}{\arg\!\min}

\begin{document}
\title{Energy Harvesting Systems with Continuous Energy and Data Arrivals: the Optimal Offline and a Heuristic Online Algorithms} %
\author{\normalsize Milad Rezaee, Mahtab~Mirmohseni and~Mohammad Reza Aref\\

Information Systems and Security Lab (ISSL)\\

Department of Electrical Engineering, Sharif University of Technology, Tehran, Iran\\

Email: miladrezaee@ee.sharif.edu,\{mirmohseni,aref\}@sharif.edu \vspace{0ex}

\thanks{The material in this paper has been accepted in part in special sessions in the ISWCS, Brussels, Belgium, Aug. 2015.}}

\markboth{}%
{}
\maketitle

\begin{abstract}
\boldmath
Energy harvesting has been developed as an effective technology for communication systems in order to extend the lifetime of these systems. In this work, we consider a single-user energy harvesting wireless communication system, in which arrival data and harvested energy curves are modeled as \emph{continuous} functions. For the single-user model, our first goal is to find an offline algorithm, which maximizes the amount of data which is transmitted to the receiver node by a given deadline. If more than one scheme exists that transmits the maximum data, we choose the one with minimum utilized energy at the transmitter node. Next, we propose an online algorithm for this system. We also consider a \emph{multi-hop} energy harvesting wireless communication system in a full-duplex mode and find the optimal offline algorithm to maximize the throughput.
\vspace{0ex}
\end{abstract}


\IEEEpeerreviewmaketitle
\section{Introduction}
Energy Harvesting (EH) has been appeared as an approach in order to make the green communications possible. In EH systems, nodes extract energy from the nature to extend their lifetimes. The harvested energy can also be used for the purpose of communication and specially for the transmission process. Compared to the conventional battery-powered systems, the EH systems have access to an unbounded source of energy (like vibration absorption devices, water mills, wind turbines, microbial fuel cells, solar cells, thermo-electric generators, piezoelectric cells, etc). However, the diffused nature of this energy makes it difficult to be used for communication. From information-theoretic point of view, the capacity of channels with EH nodes has been investigated. The basic results on the capacity of EH systems have been presented in \cite{ozel2010information}, which are continued in other works such as \cite{rajesh2011information,dong2014approximate,tutuncuoglu2014improved}.

Another important research field in this area has focused on the optimal transmission scheduling considering different optimization problems. One of these problems looks for the optimum schemes to maximize the \emph{throughput} in a given deadline \cite{ozel2011transmission,devillers2012general,tutuncuoglu2012optimum}. Most of the existing works have considered the EH nodes with \emph{discrete} energy and/or data arrivals. This problem for a single-user fading channel with additive white Gaussian noise is considered in \cite{ozel2011transmission}, where some optimal and suboptimal algorithms have been proposed. The authors in \cite{devillers2012general} consider a single-user communication with battery imperfections while the harvested energy curve is continuous. Throughput is maximized in \cite{tutuncuoglu2012optimum} while the battery is assumed to be limited. Moreover, two throughput maximization problems for single-user and multiple access channels are investigated in \cite{arafa2014single} with EH transmitters and receiver (Rx) while the Rx utilizes the harvested energy for decoding process.

Minimizing the \emph{completion time} to transmit a given amount of data is another problem which is considered by the researchers \cite{yang2012optimal,yang2012broadcasting,nagda2014optimal}. In \cite{yang2012optimal}, an algorithm is proposed to minimize the transmission period for a specific amount of given data. A broadcast channel is considered in \cite{yang2012broadcasting}, where the goal is to investigate the minimum of the transmission completion time. A single-user communication system is considered in \cite{nagda2014optimal}, in which the Rx is not supplied by an external source and its energy is provided by harvesting, resulting in a limited energy at the Rx. One of the very interesting problems in the EH systems is to consider the scenarios where some relays participate in the transmission process. A two-hop EH system is considered for the discrete energy in \cite{gurakan2013energy,luo2013optimal,gunduz2011two,orhan2012optimal}. The authors in \cite{gurakan2013energy} investigate the two-hop relay channel with EH transmitter (Tx) and relay, and one-way energy transfer from Tx to the relay. This problems in a half-duplex two-hop relay channel with EH only at the source have been considered in \cite{luo2013optimal}. In \cite{gunduz2011two}, it has been assumed that the relay and the source harvest energy from the environment and the problem in both full-duplex and half-duplex cases is investigated. \cite{orhan2012optimal} only develop the optimal offline algorithm for a half-duplex mode in a throughput maximization problem when both source and relay have only two energy packets before the given deadline. In \cite{gurakan2014energy}, the authors consider a diamond channel with one-way energy transfer from Tx node to the relays node. Moreover, there are some research works on Gaussian relay channel with direct link with EH at both of Tx and the relay, such as \cite{huang2013throughput} and \cite{feghhi2013optimal}.

Most of the above works focuses on offline algorithms where the arrival process information is provided non-causally to the EH Tx and online algorithms are less investigated; two examples are \cite{vaze2014dynamic} and \cite{vaze2013competitive}. \cite{vaze2014dynamic} considers the design of an online algorithm to maximize the throughput of a wireless communication channel with arbitrary fading coefficients. In \cite{vaze2013competitive}, the author finds a lower bound and an upper bound for the completion time of optimal online algorithm to the completion time of optimal offline algorithm, in order to transfer a given data in single-user and multiple-access channels with Gaussian additive noise and EH nodes. In \cite{ozel2011transmission,tutuncuoglu2012optimum} and \cite{xu2014throughput} some heuristic online schemes are proposed. The stochastic nature of the harvesting processes are taken account in \cite{lei2009generic,mao2012near,wang2014power} where optimal transmission policies for EH nodes based on Markov decision processes are studied.

As mentioned above, \emph{discrete} energy and/or data arrivals assumption is used in the most of the existing works.
Since the harvested energy in an energy harvester is naturally continuous by time, a continuous-time model is more suitable for the amount of harvested energy \cite{varan2014energy,ottman2002adaptive}. Although, assuming a discrete model makes the problem more tractable, the resulted optimal policy for such a model is a suboptimal policy in general and it reduces the efficiency. One motivation for the continuous data arrival comes from the relaying structure. In general, in a throughput maximization problem at a multi-hop relay channel, the arrival data curve at the relay node may be continuous even when the data arrival is discrete at Tx. In Section V, we provide an example to show that a throughput maximization problem in a two-hop channel with continuous energy arrivals in Tx and the relay and buffered data in Tx reduces to a throughput maximization problem in a single-user channel with continuous data arrival and harvested energy in the relay. Moreover, using rateless codes eliminates the need of packetizing data in some applications \cite{palanki2004rateless} that results in the continuous model better fitting these cases.
Even if we consider packetized communications, since in a network there are huge number of arrival packets with different sizes and arrival times which are sufficiently small and close to each other, the continuous model is more suitable in network calculus \cite{le2001network}. Therefore, investigating a system with continuous data arrival curve is crucial in analyzing EH systems. By considering the continuous energy and data arrivals, the problem enters a new space where the existing discrete-space proofs are not applicable. Noting this fact, the central question is how the existing discrete-space results change in this new continuous-space. We answer this question in this paper by providing the proofs which fit the continuous-space. Note that if we assume that arrival data is discrete, due to the continuous harvested energy, the problem enters in continuous space and the complexity of the proofs and the optimal algorithms do not change much compared with the case where both data and energy are continuous. We remark that, to the best of our knowledge, the model with discrete data arrival (not buffered) and continuous energy has not been considered in the previous works.
In addition, there are very limited results for the EH systems with continuous energy arrivals \cite{devillers2012general}. The authors in \cite{varan2014energy} investigate an EH system with a degrading battery of finite capacity by convex analysis tools for a continuous harvested energy curve. Although, in \cite{devillers2012general} and \cite{varan2014energy}, the harvested energy curve is continuous but all data is stored in information buffer at the beginning of the transmission and the arrival data curve has not been considered. Also, \cite{varan2014energy} considers finite batteries, battery imperfections, and processing costs.

The most challenging parts in this paper are to apply data and energy causalities in continuous space, which need totally different approaches from the discrete model in \cite{yang2012optimal} and the continuous model in
\cite{varan2014energy}, in which only harvested energy is a continuous function and arrival data has not been considered. Another difference between our work and \cite{yang2012optimal} is that: \cite{yang2012optimal} finds optimal policy between piecewise linear functions for the discrete harvested energy and arrival data curves; however, considering our continuous energy and data arrival curves, we search among the set of general functions (detailed in Section II). Our method for finding the algorithm considers both continuous and discrete arrival data functions as well as both continuous and discrete harvested energy functions. Moreover, there is a basic difference: \cite{yang2012optimal} investigates the dual problem i.e., a \emph{completion time minimization problem}, while we investigate a throughput maximization problem.

In this paper, we consider a single-user communication channel with an EH Tx with random data arrivals. We model harvested energy curve and arrival data curve with continuous functions in time and assume that the size of the energy and data buffers at Tx and Rx are infinite. We focus on a throughput maximization problem. However, it is possible that there exist more than one scheme which maximize the throughput. This happens when the harvested energy is more than the needed energy to transmit all of arrival data until the deadline (depends on the harvested energy curve). Hence, we need an extra condition for our model. We apply a constraint on the utilized energy to be minimum when this situation occurs. For this setup, we investigate the optimal policy which maximizes the amount of data
transmitted to Rx by a given deadline. If more than one scheme exist that transmit the maximum data, we choose the one with minimum utilized energy at the Tx.

First, we consider the optimal offline policy and state its properties in two fold: we use new proof techniques to extend the discrete case properties of \cite{yang2012optimal} and then we prove novel
properties for our continuous model which also hold for the existing discrete models. Then, based on these results, we propose an offline algorithm and show its optimality. We also consider a multi-hop
channel with one Tx, one Rx and many relays. We propose the optimal offline algorithms in both throughput maximization and completion time minimization problems in a full-duplex mode.
However, in practice, we may have no information about the future of the harvested energy and data arrival in Tx. Thus, we need an algorithm which without knowing the future amounts of harvested energy and data arrival determines the power in Tx. For this reason, we propose an online algorithm in Section VI where we do not have access to any information about the distribution of the
harvesting and arrival processes. We prove that our online algorithm uses all of the energy or sends all of the data in the data buffer and the transmitted power curve is a nondecreasing function,
similar to the optimal offline algorithm. Then, we derive a lower bound on the ratio of the amount of transmitted data in the online algorithm to the optimal offline algorithm and we show the cases
where the ratio is good enough (e.g., more than 0.5 or equal to 1). Our proposed online algorithm reduces to the optimal online algorithm for the discrete harvested energy case with no data arrival, derived in
 \cite{vaze2014dynamic}. In addition, we provide two sets of examples to assess our results numerically. Also, we show that discretizing harvested energy reduce the efficiency of system.

 The rest of the paper is organized as follow. In Section II, we formulate the main problem by an optimization problem. In Section III, we state the properties of the optimal energy and data transmitted curves for the optimal offline algorithm. In Section IV, we first prove a theorem for a simpler special case of the main problem; then, we propose the offline algorithm which gives us the optimal transmitted data curve. In Section V, we investigate a multi-hop channel by a throughput maximization problem. In Section VI, we propose an online algorithm for optimization problem of Section II, and in Section VII, we provide the simulation results. Finally, Section VIII concludes the paper.
 \theoremstyle{definition}
 \newtheorem{defn}{Definition}[section]
\section{System Model}
We consider a single-user wireless communication system, where the Tx is a node that harvests energy from an external source in a continuous fashion. We assume that at time $t$, we have $B_{s}\left (t \right )$ bits of data
 (where $B_{s}(t)$ is a continuous function) available at the Tx. The Rx is assumed to have enough energy to provide adequate power for decoding at any rate that can be achieved by the Tx. Also, we have the following assumptions.

  The harvested energy curve ($E_{s}(t)$) and arrival data curve ($B_{s}(t)$) are bounded differentiable functions of time $t$, for $t \in [0,\infty)$ (except probably in finite points because of discontinuity in these points or not equal right and left derivatives) which denote the amount of energy harvested and the amount of data arrived at the Tx in the interval $[0,t]$ respectively. Also we assume that derivative of $B_{s}(t)$ and $E_{s}(t)$ are piecewise continuous. The transmitted energy curve ($E(t)$) and the transmitted data curve ($B(t)$) are continuous functions. We assume that $E(t)$ and $B(t)$ are differentiable functions of time, $t$, for $t \in [0,\infty)$ (except probably in finite points) and these denote the amount of energy utilized and the amount of data that transmitted at the Tx in the interval $[0,t]$ for $t \in [0,\infty)$. The transmitted power curve $p(t)$ is a piecewise continuous function, that denotes the amount of power used at the Tx for $t \in [0,\infty)$. Also, we use $B^{*}(t)$, $E^{*}(t)$ and $p^{*}(t)$ as the optimal transmitted data, energy and power curves respectively.

   We assume that the instantaneous transmission rate relates
   to the power of transmission through a
   continuous function $r(p)$, which satisfies the following properties: i) $r(0)=0$, ii) $r(p)$ is a non-negative strictly concave function in $p$, iii) $r(p)$ is differentiable, iv) $r(p)$ increases monotonically in $p$, and v) $\lim_{p \to\infty }r(p)=\infty$.

   It can easily  be seen that the above conditions are satisfied in many
   systems with practical encoding/decoding schemes, such as single-user additive white Gaussian
   noise channel with optimal random coding, i.e.,  $r(p)=\frac{1}{2}\log(1+p)$.

   According to the above assumptions, we have:
\vspace*{-.1cm}
  \begin{align} \label{m}
   B(t)=\int_{0}^{t}r(\dfrac{d}{dt^{'}}E({t}'))d{t}',
      \end{align}
      In Sections III and IV, our goal is to find an offline algorithm, which maximizes the amount of data transmitted to the Rx in a given deadline. If more than one scheme exist which maximize the transmitted data, we impose another constraint: the algorithm must give the scheme, in which the amount of data is maximized, while the utilized energy is minimized at the Tx. Therefore, the optimization problem is:
   \vspace*{-.2cm}
       \begin{align} \label{kjhl}
          D(T)=&\max_{p(t)} \int_{0}^{T}r(p(t))dt\\
         & s.t.~ \int_{0}^{t}p({t}')d{t}'\leq E_{s}(t),~ 0\leq t\leq T\label{kj}\\
         &\int_{0}^{t}r(p({t}'))d{t}'\leq B_{s}(t),~0\leq t\leq T,\label{gh}
            \end{align}
            and if there exist more than one $p(t)$ such that $D(T)=B_{s}(T)$, then the one is selected which uses the minimum energy.

\eqref{kj} and \eqref{gh} hold due to the causality of energy and data, respectively. The causality of energy means that one cannot use the energy which is not harvested and the causality of data means that the data which has not arrived yet, cannot be sent. We remark that we investigate both continuous and discrete harvested energy and arrival data curves in our model.
\section{Properties of the Optimal Policy}
In this section, we state the properties of the optimal policy by some useful lemmas. In Subsection~\ref{subs:existing}, we extend the discrete case properties of \cite{yang2012optimal} to our model which needs new proof techniques. In Subsection~\ref{subs:new}, we prove new properties which hold for both continuous and discrete cases.
 \newtheorem{theorem}{Theorem}[section]
   \newtheorem{corollary}{Corollary}[theorem]
   \newtheorem{lemma}[theorem]{Lemma}

\subsection{Extending the discrete case properties to our model}\label{subs:existing}
\begin{lemma}\label{Jensen}
If $E(t)$ is  nonlinear in $t$ over $[a,b]$, the straight line which passes through the two points, $(a,E(a))$ and $(b,E(b))$, transmits more data than $E(t)$ (utilizing same amount of energy). Similarly, if $B(t)$ is  nonlinear in $t$ over $[a,b]$, the straight line which passes through the two points, $(a,B(a))$ and $(b,B(b))$, utilizes less energy than $B(t)$ (transmitting same amount of data).
\end{lemma}

\begin{proof}
With using Jensen's inequality in \cite{jeffrey2007table}
and substitute $f(.)=p(.)$ and $\phi (.)=r(.)$. Assuming that $p(.)$ is not a constant function in $[a,b]$ and $a\neq b$, because of $r(.)$ is strictly concave we obtain:
                       \begin{align}
                          r\left ( \frac{E(b)-E(a)}{b-a} \right )(b-a)> \int_{a}^{b}r(p(t))dt.
                                        \end{align}
Similarly, the result can be shown for $B(t)$. This completes the proof.
\end{proof}

\begin{lemma}\label{fci} Let $E(t)$ be a feasible transmitted energy curve, $m(t)$ be a straight line fragment over $[a,b]$ that passes through points $(a,E(a))$ and $(b,E(b))$, then $E(t)$ is not the optimal transmitted energy curve if  there exist a curve $E_{new}(t) \not \equiv E(t)$ which satisfies the causality conditions \eqref{kj} and \eqref{gh}. Also
    \begin{align} \label{addd}
                              E_{new}(t)= \left\{\begin{matrix}
                                        E(t) &0\leq t< a \\
                                        m(t) &a\leq t\leq b \\
                                        E(t) &b<t\leq T
                                        \end{matrix}\right..
                                         \end{align}
\end{lemma}
\begin{proof} We have to show that $E_{new}(t)$ transmits more data than $E(t)$ while it uses the same energy. First, it is clear that $E_{new}(t)$ and $E(t)$ use equal energy to send equal amount of data in $(0,a)$ and $(b,T)$. \eqref{addd} implies that $E(t)$ transmits smaller amount of data than $E_{new}(t)$, using the same energy in $(a,b)$. Thus, overall $E_{new}(t)$ transmits more amount of data and  $E(t)$ is not optimal.\end{proof}

 Similarly, it can be shown that Lemma \ref{fci} is valid if we replace energy transmitted curve with data transmitted curve. Using the above lemma, we conclude that in the optimal policy, there are no two points on $E^{*}(t)$ such that the line passing through these points satisfies the causality conditions and is not equal to $E^{*}(t)$. The same results holds for $B^{*}(t)$.

\begin{lemma}\label{sdddf}
$p^{*}(t)$ is not decreasing.
\end{lemma}
\begin{proof} We prove this lemma by contradiction. Let there exists an interval where $p^{*}(t)$ is strictly decreasing. We propose a transmitted energy curve ($E(t)$) such that transmitted data curve ($B(t)$) is equal to $B^{*}(t)$ at $t=T$, while $E(t)$ uses smaller energy than $E^{*}(t)$.
Based on contradiction assumption, we assume that $p^{*}(t)$ is strictly decreasing in $[t_{c},t_{c}+\delta]$, then $E^{*}(t)$ is strictly concave in this interval. \eqref{addd} implies that the straight line which passes through $t_{c}$ and $t_{c}+\delta$ transmits more amount of data than $E^{*}(t)$ in this interval. Now, we decrease the slope of this line to $p_{0}$ such that this line transmits equal data with $E^{*}(t)$ in $[t_{c},t_{c}+\delta]$.
We define $E(t)$ a curve such that: (i) it is equal to $E^{*}(t)$ except at $[t_{c},T]$ ; (ii) it is a line with slope of $p_{0}$, that proposed above, at $ [t_{c}, t_{c}+\delta]$; (iii) it is equal to $E^{*}(t)-\epsilon$ where $\epsilon=E^{*}(t_{c}+\delta)-E(t_{c}+\delta)$ at $(t_{c} +\delta , T]$. Since, $E^{*}(t)$ and $E(t)$ transmits same amount of data in $[t_{c}, t_{c}+\delta]$, there exists $t_{0}\in (t_{c},t_{c}+\delta)$ such that $p^{*}(t)>p_{0}$ for $t_{c}<t<t_{0}$, and $p^{*}(t)<p_{0}$ in $t_{0}<t<t_{c}+\delta$. Otherwise, $p^{*}(t)>p_{0}$ or $p^{*}(t)<p_{0}$, which both of them are contradictions: because they result in $\int_{t_{c}}^{t_{c}+\delta }r(p^{*}(\tau))d\tau <\delta r(p_{0})$ or $\int_{t_{c}}^{t_{c}+\delta }r(p^{*}(\tau))d\tau >\delta r(p_{0})$, but
    \begin{align}
    \label{ggg} \int_{t_{c}}^{t_{c}+\delta }r(p^{*}(\tau))d\tau =\delta r(p_{0})
    \end{align} must hold.

     It is enough to prove that the transmitted data curve that obtains from $E(t)$, i.e., $B(t)$, satisfies the data causality condition. In $[0,t_{c})$, since $E^{*}(t)=E(t)$, then $B^{*}(t)=B(t)$. Clearly, if $B(t)\leq B^{*}(t)$ in $[t_{c},t_{c}+\delta]$ then $B(t)\leq B^{*}(t)$ in $[t_{c}+\delta,T]$ which means that $B(t)$ satisfies causality of data. Hence, we show that $B(t)\leq B^{*}(t)$ in $[t_{c},t_{c}+\delta]$.
    Since $p^{*}(x)> p_{0}$ and $r(.)$ increases monotonically, we have $\int_{t_{c}}^{x }r(p^{*}(\tau))d\tau >(x-t_{c}) r(p_{0})$ for $t_{c}<x<t_{0}$. Now, we use contradiction to prove that $B(t)\leq B^{*}(t)$ in $[t_{0},t_{c}+\delta]$. To do this we assume that, there exists a point $a$, such that $t_{0}\leq a\leq t_{c}+\delta$ and $B(a)> B^{*}(a)$. Hence,
    \begin{align}
    \int_{t_{c}}^{a}r(p^{*}(\tau))d\tau<(a-t_{c})r(p_{0}).
    \end{align}
which using \eqref{ggg} we obtain:
    \begin{align}\label{aaa}
    \int_{t_{c}}^{t_{c}+\delta }r(p^{*}(\tau))d\tau-\int_{t_{c}}^{a}r(p^{*}(\tau))d\tau>\delta r(p_{0})-(a-t_{c})r(p_{0}) \nonumber\\
    \Rightarrow \int_{a}^{t_{c}+\delta }r(p^{*}(\tau ))d\tau >(t_{c}+\delta-a)r(p_{0}).~~~~~~
    \end{align}
     However, for every $x\in [a,t_{c}+\delta]$, we have $p^{*}(x)\leq p_{0}$ and since $r(.)$ increases monotonically , we obtain
     \begin{align}
     \int_{a}^{t_{c}+\delta }r(p^{*}(\tau ))d\tau \leq (t_{c}+\delta-a)r(p_{0}).
     \end{align}
     which is inconsistent with \eqref{aaa}. Thus, $E(T)<E^{*}(T)$, but $B(T)=B^{*}(T)$. Now, if $B(T)=B_{s}(T)$, then we have contradiction due to the assumption that $p^{*}(t)$ is optimal and proof is completed. Otherwise, if $B(T)<B_{s}(T)$, we can increase $p(t)$ slightly in $(T-\epsilon_{1},T)$, which implies that $p^{*}(t)$ is not optimal. This completes the proof. \end{proof}
                \begin{lemma}\label{jgh}  Under the optimal policy, if there exists an epoch that no energy and no data are received, i.e., if $B_{s}(t)$ and $E_{s}(t)$ are constant, then $p^{*}(t)$ is constant in this epoch.
                \end{lemma}
\newtheorem{conj}{Conjecture}
\newtheorem{rem}{Remark}
\newtheorem{coro}{Corollary}
  \begin{proof} Similar to the above lemma, we prove this lemma by proposing a suitable transmitted energy curve $E(t)$. We assume that $p^{*}(t)$ is not constant while $E_{s}(t)$ and $B_{s}(t)$ are constant in $[a,b]$. Similar to the proof of the above lemma, there exists a point $c$ such that, we can replace  $E^{*}(t)$ in $(a,c)$ with a straight line, where the new curves satisfy both data and energy causality conditions, $E(T)<E^{*}(T)$ and $B(T)=B^{*}(T)$. Hence, $E^{*}(t)$ is not optimal. For brevity, we do not include the proof details.
  \end{proof}

\begin{lemma}\label{sfp}  Under the optimal policy, whenever $p^{*}(t)$ increases (at instant $t_{0}$), at least one of the followings holds: (i) $E^{*}(t_{0})=E_{s}(t_{0})$, (ii) $B^{*}(t_{0})=B_{s}(t_{0})$.
                \end{lemma}

      \begin{proof} The proof is similar to the proof of Lemma \ref{jgh} by proposing an appropriate  transmitted energy curve $E(t)$.
      \end{proof}

\begin{coro} Based on Lemmas \ref{jgh} and \ref{sfp}, the optimal transmitted data/energy curve must be linear except probably in the epochs in which $E^{*}(t)$ equals $E_{s}(t)$ or $B^{*}(t)$ equals $B_{s}(t)$. Also based on Lemma \ref{sdddf} , since $p^{*}(t)$ is an increasing function, $E^{*}(t)$ and $B^{*}(t)$ are convex functions.
         \end{coro}
\begin{rem}
Though similar results to the ones in Lemmas \ref{sdddf}, \ref{jgh} and \ref{sfp} have been proposed in \cite{yang2012optimal} for discrete energy and data arrivals for the completion time minimization problem, we show these results for the \emph{continuous} case using different proof techniques from the proofs in \cite{yang2012optimal}, noting that the proofs of \cite{yang2012optimal} cannot easily be extended to the continuous model.
\end{rem}

\subsection{New general properties}\label{subs:new}
\begin{lemma}\label{lemma6} Assume that $f(t)$ and $g(t)\not\equiv f(t)$ are continuous piecewise differentiable functions in $[a,b]$ such that $g(t)\leq f(t)$ in $(a,b)$ and $f(a)=g(a)$ and $f(b)=g(b)$. If $f(t)$ is convex and increasing in $[a,b]$, then $len_{[a,b]}(f(t))<len_{[a,b]}(g(t))$, where $len_{[a,b]}(f(t))$ means the length of curve $f(t)$ in interval $[a,b]$.
\end{lemma}
\begin{figure}\label{fig1}
  \centering
  \includegraphics[width=3.3in]{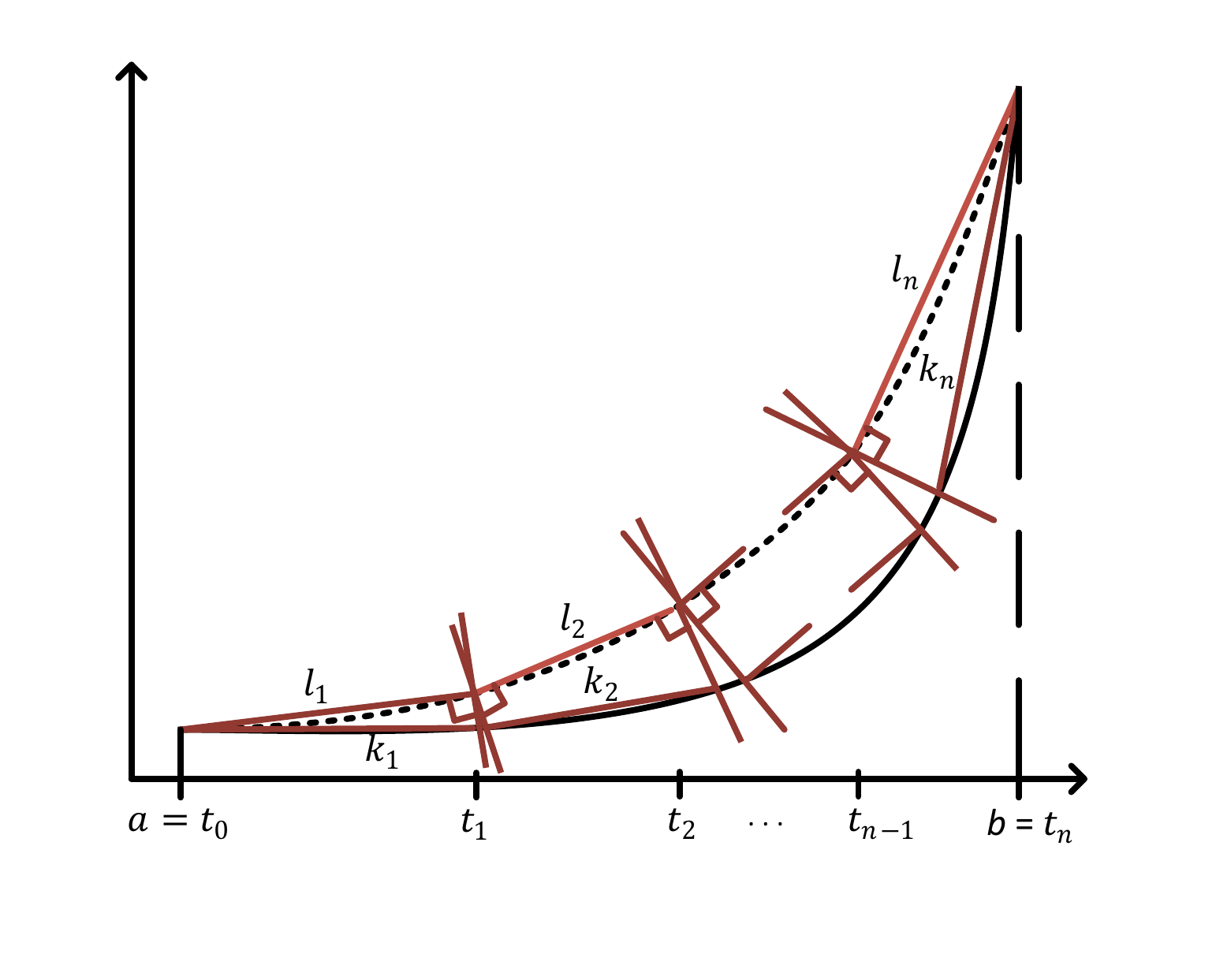}
  \caption{The doted curve $f(t)$, the continuous curve $g(t)$}
  \end{figure}

  \begin{figure}
    \centering
    \includegraphics[width=3.3in]{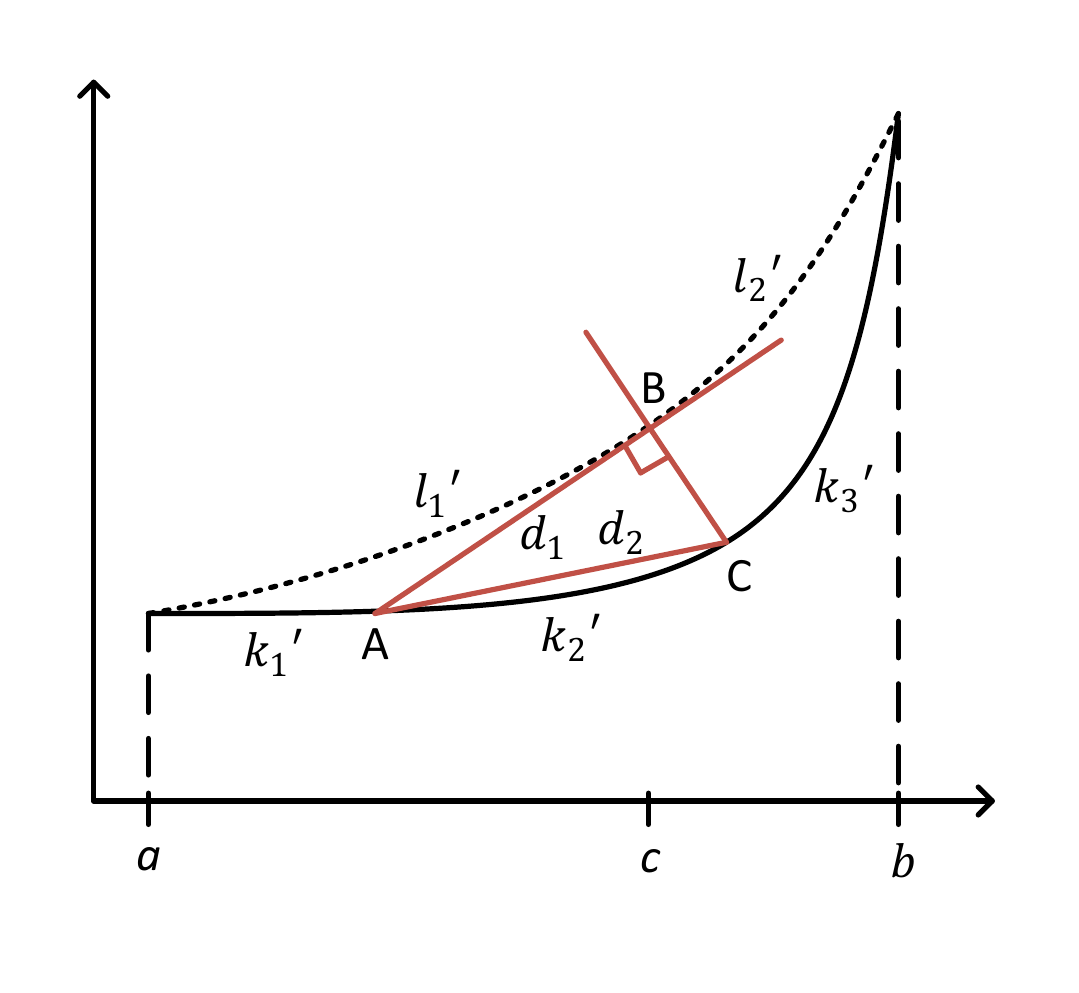}
    \caption{The doted curve $f(t)$, the continuous curve $g(t)$}
    \vspace*{-.5cm}
    \end{figure}

\begin{proof} As illustrated in Fig. 1, we choose arbitrarily points $t_{1},~t_{2},...,~t_{n-1}$ in interval $(a,b)$ and we draw the lines that connect any two adjacent points on $f(t)$ and we denote these line segments by $l_{1},~l_{2},...,~l_{n}$. Then, we draw the perpendicular lines on $l_{i}$ at both ends to collide $g(t)$ and draw line segments $k_{1},~k_{2},...,~k_{n}$ on $g(t)$ through the adjacent points of collisions.
Also assume that curves corresponding to $l_{i}$ and $k_{i}$ is $l_{i}^{'}$ and $k_{i}^{'}$, respectively on $f(t)$ and $g(t)$ in Fig. 1. According to Fig. 1, $\forall i$: $len(l_{i})\leq len(k_{i})\leq len(k_{i}^{'})$, where $len$ is calculated in the defined interval. Thus, we can write,
\begin{align}\label{convexshort}
\sum_{i=1}^{n}len(l_{i})\leq \sum_{i=1}^{n}len(k_{i})\leq \sum_{i=1}^{n}len(k_{i}^{'})\leq len_{[a,b]}(g(t)).
\end{align}
Since for all partitions $t_{1},~t_{2},...,~t_{n-1}$ in interval $(a,b)$, inequality (\ref{convexshort}) holds, we have:
\begin{align}\label{10}
len_{[a,b]}(f(t))=\sup_{I} \sum len(l_{i})\leq len_{[a,b]}(g(t)),
\end{align}
where $I$ is the set of all partitions in $(a,b)$.

Since $g(t)\not\equiv f(t)$, there exists $c$ such that $g(c)<f(c)$. We draw vertical line and tangent to $f(t)$ in $c$ as Fig. 2. If we assume that $l_{1}^{'}=f_{1}(t)$ and define $g_{1}(t)$ as the composition of two curves $k_{1}^{'}$ and $d_{1}$, then from (\ref{10}) we have,

\begin{align}
len(l_{1}^{'})\leq len(d_{1})+len(k_{1}^{'}).
\end{align}
On the other hand, since $d_{2}$ is hypotenuse in right triangle $ABC$, we obtain $d_{1}<d_{2}$. Hence,
\begin{align}\label{16}
len(l_{1}^{'})\leq len(d_{1})+len(k_{1}^{'}) <len(d_{2})+len(k_{1}^{'})\leq \\ \nonumber
len(k_{2}^{'})+len(k_{1}^{'}).~~~~~~~~~~~~~~~~~~
\end{align}
Similarly we can prove,
\begin{align}\label{17}
len(l_{2}^{'})<len(k_{3}^{'}).
\end{align}
(\ref{16}) and (\ref{17}) conclude,
\begin{align}
len(l_{1}^{'})+len(l_{2}^{'})\!\!<\!\!len(k_{1}^{'})+len(k_{2}^{'})+len(k_{3}^{'})\!\!\Rightarrow\\ \nonumber len_{[a,b]}(f(t))\!\!<\!\!len_{[a,b]}(g(t)).~~~~~~~~~~~
\end{align} \end{proof}

\begin{lemma}\label{III.7.}If in an interval we have $\dfrac{d}{dt}p^{*}(t)\neq 0$, then $B(t)\leq B^{*}(t)$.
\end{lemma}

\begin{proof} To prove this lemma we use contradiction. Assume that $t_{0}$ is a point such that $\dfrac{d}{dt}p^{*}(t_{0})\neq 0$ and there exists a transmitted data curve $B(t)$ that $B^{*}(t_{0})<B(t_{0})$ which results in $B^{*}(t_{0})< B_{s}(t_{0})$. Since $\dfrac{d}{dt}p^{*}(t_{0})\neq 0$, from Lemma \ref{sfp}, we have $E^{*}(t_{0})=E_{s}(t_{0})$ which means $E(t_{0})\leq E^{*}(t_{0})$. Consider the transmitted power curve $p_{1}(t)$ as below,

\begin{align}\label{p1}
p_{1}(t)=\left\{\begin{matrix}
p(t) & 0\leq t<t_{0} \\
0 & t_{0}\leq t<t_{0}+\epsilon \\
p^{*}(t) & t_{0}+\epsilon\leq t\leq T
\end{matrix}\right..
\end{align}
where $p(t)$ corresponds to the transmitted data curve $B(t)$. Now we show that $p_{1}(t)$ is more efficient than $p^{*}(t)$, thus we find an $\epsilon$ such that $0<\epsilon\leq T-t_{0}$ in (\ref{p1}) such that,
\begin{align}\label{20}
\int_{t_{0}}^{t_{0}+\epsilon}r(p^{*}(t))dt=B(t_{0})-B^{*}(t_{0}).
\end{align}
If there exists an $\epsilon$ that satisfies (\ref{20}), then $p_{1}(t)$ transmits $B_{1}(T)=B^{*}(T)$ data using $E_{1}(T)<E^{*}(T)$ energy and satisfies causality conditions: this means that we have a contradiction, and if there does not exist any $\epsilon$ in interval $(0,T-t_{0}]$, we assume that $p_{1}(t)$ is as follows,
\begin{align}
p_{1}(t)=\left\{\begin{matrix}
p(t) & 0\leq t<t_{0} \\
0 & t_{0}\leq t\leq T
\end{matrix}\right..
\end{align}
 From above, which results in $B^{*}(T)<B_{1}(t_{0})=B_{1}(T)$: this means we have a contradiction, too.\end{proof}

\begin{theorem}\label{III.9} If in the optimal policy , the data $B_{s}(T)$ is totally transmitted, i.e., $B_{s}(T)=B^{*}(T)$, then curve $B^{*}(t)$ has minimum length among the feasible transmitted data curves that transmit all of $B_{s}(T)$, that is, $B^{*}(t)$ minimizes the metric,
\begin{align}
len_{[0,T]}(B(t))=\int_{0}^{T}\sqrt{1+(\dfrac{d}{dt}B(t))^{2}}~dt
\end{align}among feasible curves which connects origin to $(T,B_{s}(T))$.
\end{theorem}

\begin{proof} Assume that $B(t)$ is a feasible curve such that $B(T)=B_{s}(T)$ and $B(t)\not\equiv B^{*}(t)$. Based on Lemma \ref{III.7.}, whenever in an interval $B(t)>B^{*}(t)$, then $B^{*}(t)$ must be linear in this interval. We divide the interval $[0,T]$ as follow:

1- Intervals $[c_{i},d_{i}]$, $1\leq i\leq n$, in which we have $B(t)\leq B^{*}(t)$, $B(c_{i})= B^{*}(c_{i})$ and $B(d_{i})= B^{*}(d_{i})$, $\forall i$.

2-Intervals $[e_{i},f_{i}]$ , $1\leq i\leq m$, in which we have $B^{*}(t)< B(t)$ in  $(e_{i},f_{i})$, $B(e_{i})= B^{*}(e_{i})$ and $B(f_{i})= B^{*}(f_{i})$, $\forall i$.

In the intervals of part 1 based on Lemma \ref{lemma6}, $len_{[c_{i},d_{i}]}B^{*}(t)\leq len_{[c_{i},d_{i}]}B(t)$ and if $B^{*}(t)\not\equiv B(t)$ in $[c_{i},d_{i}]$ then $len_{[c_{i},d_{i}]}B^{*}(t)< len_{[c_{i},d_{i}]}B(t)$.

In the intervals of part 2, as declared above, $B^{*}(t)$ in intervals $[e_{i},f_{i}]$, $\forall 1\leq i\leq m$ is linear and it conclude
$len_{[e_{i},f_{i}]}B^{*}(t)< len_{[e_{i},f_{i}]}B(t)$.
Hence, we have $len_{[c_{i},d_{i}]}B^{*}(t)< len_{[c_{i},d_{i}]}B(t)$.\end{proof}

\begin{conj}\label{III.8} If in the optimal policy $E_{s}(T)$ is totally used, i.e., $E_{s}(T)=E^{*}(T)$, then, curve $E^{*}(t)$ has minimum length among the feasible transmitted energy curves that use $E_{s}(T)$ totally, that is, $E^{*}(t)$ minimizes the metric,
\begin{align}
len_{[0,T]}(E(t))=\int_{0}^{T}\sqrt{1+(\dfrac{d}{dt}E(t))^{2}}dt.
\end{align}
\end{conj}

\begin{theorem}\label{III.11} In the optimal policy we have,
\begin{align}
\frac{d}{dt}B^{*}(t)\rvert_{t=T^{-}}\leq \max_{[0,T]} \frac{d}{dt}B(t),
\end{align}
where $B(t)$ is any arbitrary feasible transmitted data curve which connects the origin to $(T,B^{*}(T))$.
\end{theorem}

\begin{proof}If $\frac{d}{dt}B^{*}(t)\rvert_{t=T^{-}}\leq \frac{d}{dt}B(t)\rvert_{t=T^{-}}$ then proof is complete. Hence we assume that $\frac{d}{dt}B(t)\rvert_{t=T^{-}}<\frac{d}{dt}B^{*}(t)\rvert_{t=T^{-}}$. Since $\frac{d}{dt}B(t)\rvert_{t=T^{-}}<\frac{d}{dt}B^{*}(t)\rvert_{t=T^{-}}$ and $B(T)=B^{*}(T)$, there exists an $\epsilon$ such that for any $t\in (T-\epsilon,T)$, $B^{*}(t)< B(t)$. Also, we assume that $\epsilon=\min\{T-t:B(t)=B^{*}(t),~ t\in [0,T)\}$. Thus based on Lemma \ref{III.7.} for $t\in [T-\epsilon,T]$, $B^{*}(t)$ is linear ($\frac{d}{dt}B^{*}(t)\rvert_{t=(T-\epsilon)^{+}}=\frac{d}{dt}B^{*}(t)\rvert_{t=T^{-}}$). Thus due to $B^{*}(t)\leq B(t)$ for $t\in [T-\epsilon,T]$ and $B^{*}(T-\epsilon)=B(T-\epsilon)$ then $\frac{d}{dt}B^{*}(t)\rvert_{t=(T-\epsilon)^{+}}\leq \frac{d}{dt}B(t)\rvert_{t=(T-\epsilon)^{+}}$. This completes the proof.\end{proof}

\begin{rem} The importance of the Conjecture \ref{III.8} and Theorem~\ref{III.9} comes from the fact that if we can prove Conjecture~\ref{III.8}, we may propose the optimal offline algorithm as the shortest path curve among all admissible policies which use all the energy or send all the data in data buffer until $T$. As a result, we have a method to describe the optimal offline algorithm.
\end{rem}
\begin{rem}
Theorem \ref{III.11} describes that the maximum transmitted instantaneous power in the optimal policy is less than the maximum transmitted instantaneous power for all feasible policies. This becomes very significant if we impose an additional maximum power constraint in the optimization problem in \eqref{kjhl}-\eqref{gh}. In this case, we first solve the problem without considering the maximum power constraint. If the optimal policy satisfies the maximum power constraint, we are done; otherwise, there is no feasible policy that can transmit the amount of data transmitted by optimal policy with no maximum power constraint. Also, neither all of data in data buffer is sent, nor all of energy until $T$ is used for the optimal policy. As a result we can determine the cases where the maximum power constraint is a limiting element.
\end{rem}

 \section{Optimal Offline Algorithm}

  In this section, we propose the optimal offline algorithm for the optimization problem \eqref{kjhl}-\eqref{gh}. First, for simplicity, we explain the main idea of this algorithm for the discrete arrival data and discrete harvested energy curves, i.e., we assume that at instants $0,t_{1}^{E},...$ Tx harvests energy in amount of $E_{0}, E_{1}, ...$ and at instants $0,t_{1}^{B},...$ the data arrives in amount of $B_{0}, B_{1}, ...$ bits.

   \begin{defn}(Event point):
   every time in which the energy is harvested or the data is arrived is an event point.
   \end{defn}

\begin{theorem}\label{IV.1}
Let $u_{i}$ be the $(i+1)-th$ event point. Assume that there exist $m-1$ event points before $T$, and also assume that $u_{m}=T,~ u_{0}=0$. Then the optimal policy structure of the transmitted rate is as follows.
\begin{align}\label{12}
                          i_{n}= \argmin_{i}\bigg \{ \min_{u_{i_{n-1}}< u_{i}\leqslant u_{m}}r( \frac{E_{s}(u_{i}^{-})-E(u_{i_{n-1}})}{u_{i}-u_{i_{n-1}}})\nonumber\\ ,\min_{u_{i_{n-1}}< u_{i}\leqslant u_{m}}\frac{B_{s}(u_{i}^{-})-B(u_{i_{n-1}})}{u_{i}-u_{i_{n-1}}}\bigg \}                  \end{align}
\begin{align}\label{13}  r_{n}={\min}\bigg \{r( \frac{E_{s}(u_{i_{n}}^{-})-E(u_{i_{n-1}})}{u_{i_{n}}-u_{i_{n-1}}}),\nonumber\\\frac{B_{s}(u_{i_{n}}^{-})-B(u_{i_{n-1}})}{u_{i_{n}}-u_{i_{n-1}}}  \bigg\}
\end{align}
\begin{align}\label{eqn:update}
                          B(u_{i_{n}})=B(u_{i_{n-1}})+(u_{i_{n}}-u_{i_{n-1}})r_{n},
\end{align}
where  $u_{i_{0}}=0$, $n\in N$, and $r_{n}$ is the transmitted rate in the interval $(u_{i_{n-1}},u_{i_{n}})$.
\end{theorem}
 \begin{rem}
 The above algorithm works as follows. First, we use \eqref{12} to find the event point which imposes the data rate bottleneck by comparing the proper slopes. The first term shows the point which imposes this constraint due to the harvested energy shortage while the second term is related to the point where the arrival data causes this constraint. Then, we use the found even point in \eqref{13} to determine the data rate (i.e., the slope), which also shows that the bottleneck is due to harvested energy or arrival data. \eqref{eqn:update} shows the transmitted data update.
 \end{rem}

 \begin{proof} We again use contradiction. First, it is concluded from Lemma \ref{jgh} that the transmitted power and rate curves must be constant between any two event points. Thus, the transmitted data and energy curves are piecewise linear and we must find the optimal curves among the piecewise linear functions. Without loss of generality, we assume that $n=1$:
 \begin{align}\label{14}  r_{1}=\min\big \{r ( \frac{E_{s}(u_{i_{1}}^{-})}{u_{i_{1}}}  ),\frac{B_{s}(u_{i_{1}}^{-})}{u_{i_{1}}}  \big \}.
                                                \end{align}
 Noting the optimal rate in interval $(0, u_{k})$ by $r_{1}^{*}$, if $r_{1}$ is not the optimal rate, based on Lemma \ref{sfp}, we must have:

  \begin{align}  r^{*}_{1}>\min\big \{r\left ( \frac{E_{s}(u_{i_{1}}^{-})}{u_{i_{1}}} \right ),\frac{B_{s}(u_{i_{1}}^{-})}{u_{i_{1}}}  \big \}.
                                                 \end{align}
  Based on Lemma \ref{sfp}, under the optimal policy, at least one of two following equations holds,
  \begin{align}  r_{1}^{*}=r\left ( \frac{E_{s}(u_{k}^{-})}{u_{k}}\right )~~~or~~~  r_{1}^{*}=\frac{B_{s}(u_{k}^{-})}{u_{k}}.
  \end{align}
We consider two cases: (i) if $u_{k}>u_{i_{1}}$, then at least one of the causality conditions is violated due to $r_{1}<r^{*}_{1}$ and \eqref{14}; (ii) if $u_{k}<u_{i_{1}}$, since the optimal policy (corresponds to $r^*$) must satisfy the constraint imposed by $r_1$ at $u_{i_1}$, the transmitted rate curve, i.e., $r^*$, must be decreased at least in an interval which is inconsistent with Lemma \ref{sdddf}. Therefore, \eqref{14} gives $r_{1}^{*}$. A similar argument proves that $r_{2},r_{3},...$, also satisfy this structure.
   \end{proof}
 \begin{rem}
 Our algorithm (tailored to the discrete curves) differs from the one in \cite{yang2012optimal} in two aspects: (i) In \cite{yang2012optimal}, the purpose is to find the best scheme among \emph{piecewise linear functions} for the transmitted data curve. However, we first prove that the optimal transmitted data curve among all functions (assumptions of section II) is a piecewise linear function. (ii) \cite{yang2012optimal} investigates a minimization completion time problem, however we investigate a throughput maximization problem. The algorithm of \cite{yang2012optimal} is as follows: it first calculates the minimum needed energy and a lower bound $T_{1}$ for completion time to transmit a given amount of data. Next, it  uses \eqref{12} and \eqref{13} from $0$ to $T_{1}$ to find the next point $u_{k}$, where it again calculates minimum needed energy and another lower bound for completion time to transmit the remaining data. Then, this procedure is repeated from $u_{k}$ to the new lower bound finally to deliver all of given data. However, we use \eqref{12} and \eqref{13} repeatedly to compare all the event points (i.e., $u_{i}:~i=0,1,...,m$). We remark that our main contribution compared to \cite{yang2012optimal} is considering the continuous model for harvested energy and arrival data curves and the above differences are the minor ones considering our general three-step algorithm tailored to the discrete curves.
   \end{rem}

 Now, we describe our proposed algorithm for the $\emph{continuous}$ model. First, we define two variables $r_{a}$ and $r_{b}$ as follow:
 \begin{align}
      r_{b}=\lim_{u \to x^{+}}r(\frac{E_{s}(u)-E(x)}{u-x}),~~~ r_{a}=\lim_{u \to x^{+}}\frac{B_{s}(u)-B(x)}{u-x}.
 \end{align}
 Our algorithm has three steps which depend on two curves $B_{s}(t)$ and $E_{s}(t)$. This algorithm is the extension of the previous algorithm that has been presented in Theorem \ref{IV.1}. The main difference of these two algorithms is that in the algorithm of Theorem \ref{IV.1}, the optimal transmitted energy/data curve is piecewise linear; while in this algorithm, in some intervals, the energy/data transmitted curve could be equal to the harvested energy curve/arrival data curve (which are continuous in general).

In this algorithm for determining the intervals in which transmitted data/energy curve is equal to arrival data/harvested energy curve, we first check some conditions; based on them, we have three states:

\textbf{State A} (linear part): The transmitted data and energy curves are straight lines.

\textbf{State B} ($E^{*}(t)=E_{s}(t)$): The transmitted energy curve is equal to harvested energy curve.

\textbf{State C} ($B^{*}(t)=B_{s}(t)$): The transmitted data curve is equal to arrival data curve.

Our algorithm works as follows. First, the conditions in Tables~I and II are checked (the details will be provided later in this section). If the conditions in Table~I (or II) hold, the algorithm enters state B (or C). Otherwise, the algorithm enters state A.

When we are in state A, the transmitted data and energy curves are straight lines and we can derive the interval of linearity. The slope of the transmitted rate is derived as:
 \begin{align}\label{31}
  \widetilde{r}=\min\bigg \{\min_{x< u\leq T}r\left ( \frac{E_{s}(u)-E(x)}{u-x} \right ) ,\nonumber\\\min_{x< u\leq T} \frac{B_{s}(u)-B(x)}{u-x} \bigg \}
   \end{align}
and the end of the linearity interval is:
\begin{align}
     \widetilde{u}= \max \bigg\{ \argmin_{u}\bigg \{\min_{x< u\leq T}r\left ( \frac{E_{s}(u)-E(x)}{u-x} \right ) ,\nonumber\\\min_{x< u\leq T} \frac{B_{s}(u)-B(x)}{u-x} \bigg \} \bigg\}
\end{align}
At the end of the interval the algorithm repeats from the beginning.

Once we are in state B or C, the corresponding conditions (in Table~I or II) must be checked continuously. The algorithm stays in these states as long as these conditions continue to hold. If one of the conditions does not hold, the algorithm starts from the beginning.

If both $r_{a}$ and $r_{b}$ are unbounded, the conditions of Tables~I and II fail to hold and the algorithm enters state A. If both $r_{a}$ and $r_{b}$ are bounded, the conditions $\si{B_1},~\si{B_2}$ and $\si{B_3}$ in Table I and $\si{C_1},~\si{C_2},~\si{C_3}$ and $\si{C_4}$ in table II must be checked. Otherwise, B4 or C5 must be checked. There are some notions used in Tables I and II defined as below.

     \begin{align} r_{c}(x)=\inf_{x< u\leq T}\bigg\{ \frac{B_{s}(u)-B(x)}{u-x}\bigg\}
       \end{align}
     \begin{align}
     r_{d}(x)=\inf_{x< u\leq T}\bigg\{ r\left(\frac{E_{s}(u)-E(x)}{u-x}\right)\bigg\}.
       \end{align}

       \begin{table}[h!]
       \centering
       \caption{}
        \begin{tabular}{|c| p{2.8in}|}
        \hline
         & Condition that enters algorithm into state B\\
        \hline \hline
        \si{B_1}& $\overline{\left\{L_{1}>E_{s}(t),x,\epsilon_{1} \right \}}, \left \{L_{2}>B_{s}(t),x,\epsilon_{2} \right \}$, $\overline{\big\{L_{1}\times E_{s}(t),x \big \}}$, $m=r(E_{s}^{'}(x))$.  \\ \hline
        \si{B_2}&  $\overline{\left\{L_{1}>E_{s}(t),x,\epsilon_{1} \right \}},\overline{ \left \{L_{2}>B_{s}(t),x,\epsilon_{2} \right \}}$, $\overline{\big\{L_{1}\times E_{s}(t),x \big \}}$, $m=r(E_{s}^{'}(x)),\big\{L_{2}\times B_{s}(t),x \big \}$. \\ \hline
        \si{B_3}&  $\overline{\left\{L_{1}>E_{s}(t),x,\epsilon_{1} \right \}},\overline{ \left \{L_{2}>B_{s}(t),x,\epsilon_{2} \right \}}$, $\overline{\big\{L_{1}\times E_{s}(t),x \big \}}$, $m=r(E_{s}^{'}(x)),\overline{\big\{L_{2}\times B_{s}(t),x \big \}}$. \\ \hline
        \si{B_4}& $r_{b}$ is bounded and $r_{a}$ is unbounded and we have  $\left \{ L_{1}>E_{s}(t),x,\epsilon \right \}~$ and $\left\{L_{1}\times E_{s}(t),x \right \}$ and $m=r_{c}(x)$.\\ \hline
        \end{tabular}
       \end{table}

        \begin{table}[h!]
              \centering
              \caption{}
               \begin{tabular}{||p{0.2in}| p{2.8in}||}
               \hline
                & Condition that enters algorithm into state C\\ [0.5ex]
               \hline \hline
               \si{C_1}&  $\left\{L_{1}>E_{s}(t),x,\epsilon_{1} \right \}$, $\overline{ \left \{L_{2}>B_{s}(t),x,\epsilon_{2} \right \}}$, $m^{'}=B_{s}^{'}(x),\overline{\big\{L_{2}\times B_{s}(t),x \big \}}$.\\ \hline
               \si{C_2}&  $\overline{\left\{L_{1}>E_{s}(t),x,\epsilon_{1} \right \}},\overline{ \left \{L_{2}>B_{s}(t),x,\epsilon_{2} \right \}},\big\{L_{1}\times E_{s}(t),x \big \}$, $m^{'}=B_{s}^{'}(x),\overline{\big\{L_{2}\times B_{s}(t),x \big \}}$. \\ \hline
               \si{C_3}&  $\overline{\left\{L_{1}>E_{s}(t),x,\epsilon_{1} \right \}}$, $\overline{ \left \{L_{2}>B_{s}(t),x,\epsilon_{2} \right \}}$, $\overline{\big\{L_{1}\times E_{s}(t),x \big \}}$, $m=r_{c}(x),\overline{\big\{L_{2}\times B_{s}(t),x \big \}}, r_{c}(x)\neq r(E_{s}^{'}(x))$. \\ \hline
               \si{C_4}& It holds none below conditions:\\
               & 1- \si{B_1}, \si{B_2}, or \si{B_3} in Table 1.\\   & 2-$(\left\{L_{1}>E_{s}(t),x,\epsilon_{1} \right \},\left\{L_{2}>B_{s}(t),x,\epsilon_{2} \right \})$\\
                 & 3- $(\left\{L_{1}>E_{s}(t),x,\epsilon_{1} \right \},\overline{ \left \{L_{2}>B_{s}(t),x,\epsilon_{2} \right \}},\big\{L_{2}\times B_{s}(t),x \big \})$\\
                 & 4- $(\overline{\left\{L_{1}>E_{s}(t),x,\epsilon_{1} \right \}},~ \left \{L_{2}>B_{s}(t),x,\epsilon_{2} \right \},\big\{L_{1}\times E_{s}(t),x \big \})$\\
                 & 5- $(\overline{\left\{L_{1}>E_{s}(t),x,\epsilon_{1} \right \}},\overline{ \left \{L_{2}>B_{s}(t),x,\epsilon_{2} \right \}},\big\{L_{1}\times E_{s}(t),x \big \},\big\{L_{2}\times B_{s}(t),x \big \})$\\
                 & 6- $(\overline{\left\{L_{1}>E_{s}(t),x,\epsilon_{1} \right \}}, \left \{L_{2}>B_{s}(t),x,\epsilon_{2} \right \},m=r_{c}(x)$.\\
                 & 7- $(\left\{L_{1}>E_{s}(t),x,\epsilon_{1} \right \}, \overline{\left \{L_{2}>B_{s}(t),x,\epsilon_{2} \right \}},m^{'}=r_{d}(x))$\\
                 & 8- $(\overline{\left\{L_{1}>E_{s}(t),x,\epsilon_{1} \right \}},\overline{ \left \{L_{2}>B_{s}(t),x,\epsilon_{2} \right \}},\big\{L_{1}\times E_{s}(t),x \big \},\overline{\big\{L_{2}\times B_{s}(t),x \big \}},m^{'}=r_{d}(x))$\\
                 & 9- $(\overline{\left\{L_{1}>E_{s}(t),x,\epsilon_{1} \right \}},\overline{ \left \{L_{2}>B_{s}(t),x,\epsilon_{2} \right \}}$, $\overline{\big\{L_{1}\times E_{s}(t),x \big \}}$, $m=r_{c}(x),\big\{L_{2}\times B_{s}(t),x \big \})$
                \\ \hline
                 \si{C_5}& $r_{a}$ is bounded and $r_{b}$ is unbounded and we have $\left \{ L_{2}>B_{s}(t),x,\epsilon \right \}$, $\left\{L_{2}\times B_{s}(t),x \right \}$ and $m^{'}=r_{d}(x)$. \\ \hline
                       \end{tabular}
                      \end{table}
       We use the notation $\left \{ L>f(t),x,\epsilon \right \}$ to show that there exists an $\epsilon$ such that the straight line $L$ is above the function $f(t)$ in $(x,x+\epsilon)$ and $\overline{ \left \{L>f(t),x,\epsilon \right \}}$ to show that there exists no $\epsilon$ such that the straight line $L$ is above the function $f(t)$ in $(x,x+\epsilon)$.
        Also, we use the notation $\left\{L\times f(t),x \right \}$ to show that the straight line $L$ collides with the function $f(t)$ for $t>x$ and  $\overline{ \left\{L\times f(t),x \right \}}$ to show that the straight line $L$ does not collide with the function $f(t)$ for $t>x$ opposite.

        In addition, assume that $L_{1}$, $L_{2}$,  are two straight lines which are tangent to curves $E_{s}(t)$ and $B_{s}(t)$ in point $x$, respectively and
        \begin{align}
          m = \min \left \{r_{c}(x), r(\frac{d}{dt}E_{s}(t))\rvert_{t=x}  \right \}              \end{align}
        \begin{align}
                 m^{'} = \min \left \{r_{d}(x), \frac{d}{dt}B_{s}(t)\rvert_{t=x}  \right \}.              \end{align}
\begin{rem}\label{rem:tech}
If we extend the discrete algorithm directly to our continuous one, the equations \eqref{12} and \eqref{13} must be calculated for every event point, which are now a continuum (every point in $t\in [0,T]$) to find the next point in order to execute the algorithm. If this new point is same as the previous point, the algorithm enters state B or C. This process repeats until $t=T$. However, we propose a set of conditions (in Tables~I and II) to determine the state in each point. Also, we can show that the transmitted data curve which obtains from the proposed algorithm is convex. For more details assume that in an interval the transmitted data curve is concave. So the algorithm is in state B or C (The transmitted energy/data curve is equal to harvested energy/arrival data curve). Consider the beginning of this interval and use the equations \eqref{12} and \eqref{13} we can conclude the slope of straight line which passes through the endpoints of this interval is less than slope of transmitted energy/data curve in the beginning of this interval (since it must follow the harvested energy/arrival data curve). This results in a contradiction.
\end{rem}

\begin{lemma}\label{zafar} Let $B_{1}(t)$ and $B_{2}(t)$ be two distinct transmitted data curves and $B_{1}(t)> B_{2}(t)$ in the interval $(a,b)$ and $B_{1}(t)=B_{2}(t)$ at $t=a$ and $t=b$. If $B_{1}(t)$ is a convex function and $B_{1}(t)$ and $B_{2}(t)$ increase monotonically in $t$, then:
\begin{align} \int_{a}^{b}r^{-1}(\dfrac{d}{dt}B_{1}(t))dt< \int_{a}^{b}r^{-1}(\frac{d}{dt}B_{2}(t))dt. \end{align}
\end{lemma}

    \begin{proof} If we assume that $A(t)=B_{1}(t)$ and $D_{min}(t)\leq B_{2}(t)$ in \cite{zafer2009calculus}, based on \cite[Theorem IV]{zafer2009calculus} it concludes that the curve which uses the minimum energy has shortest length. Also, based on Lemma \ref{lemma6}, $B_{1}(t)$ has minimum length among the feasible data transmitted curves. Thus $B_{1}(t)$ uses minimum energy and the proof is complete.\end{proof}

   \begin{lemma}\label{IV.3}In our proposed algorithm, there do not exist any two points on the transmitted data curve, $B(t)$, such that the line passing through these points satisfies both causality conditions and $B_{new}(t)\neq B(t)$, where $B_{new}(t)$ is the $B(t)$ replaced with the straight line that passes through these two points in the interval made by them.
   \end{lemma}

\begin{proof} The proof is based on contradiction. Hence, we assume that there exist two points $s$ and $l$ such $B_{new}(t)$ does not violate both causality conditions and $B_{new}(t)\neq B(t)$. As explained,$B(t)$ has at most three parts: 1- linear part, 2- some parts in which, $B(t)=B_{s}(t)$, 3- some parts in which $E(t)=E_{s}(t)$. Moreover, since the algorithm obtains a convex transmitted data curve, the straight line that passes through $s$ and $l$ is above of $B(t)$ in $(s,l)$, i.e., $B_{new}(t)>B(t)$ in $(s,l)$. It is clear that $s$ and $l$ are not on the linear part. Hence, there exists a point $x$ in $(s,l)$ such that $B(x)=B_{s}(x)$ or $E(x)=E_{s}(x)$. We assume $M=\left\{t:~B(t)=B_{s}(t)~or~E(t)=E_{s}(t),~t\in(s,l) \right\}$ and $v=\inf(M)$. If $v=s$ there exists an $\epsilon$ such that $B(t)=B_{s}(t)$ or $E(t)=E_{s}(t)$ for $s<t<s+\epsilon$. Thus one of the causality conditions is violated: because, if $B(t)=B_{s}(t)$ in $s<t<s+\epsilon$, we have $B(t)=B_{s}(t)$ and $B_{new}(t)>B(t)$ which result in $B_{new}(t)>B_{s}(t)$. Therefore, the data causality condition is violated. If $E(t)=E_{s}(t)$ in $s<t<s+\epsilon$, there exists an $\epsilon_{1}$ such that in $s<t<s+\epsilon_{1}$, we have $p(t)<p_{new}(t)$. Thus, $E(t)<E_{new}(t)$ which results in $E_{s}(t)<E_{new}(t)$ in $s<t<s+\epsilon_{1}$ and the energy causality condition violated. Hence, we have $v\neq s$ which results that $t=v$ is the first instant, in which $r(t)=\frac{d}{dt}B(t)$ can be changing in $s<t<l$. If $B(v)=B_{s}(v)$, then the data causality condition is violated. Now assume that $E(v)=E_{s}(v)$: $B(t)$ is linear in $(s,v)$ and its slope is smaller than the slope of the straight line in curve $B_{new}(t)$ in $s<t<l$. Therefore, $p(t)<p_{new}(t)$ in $(s,v)$ which results in $E(t)<E_{new}(t)$ in $(s,v)$. Thus $E_{s}(v)<E_{new}(v)$ and the causality of energy is violated at $t=v$. This results in a contradiction which completes the proof. \end{proof}

\begin{lemma}\label{IV.4}
If there exists a convex transmitted data curve, $B(t)$ such that $B(T)=B_{s}(T)$ or $E(T)=E_{s}(T)$ and satisfies the condition of Lemma~\ref{IV.3}, then $B(t)$ is optimal.
\end{lemma}
\begin{proof} Again we use contradiction and we assume that $B(t)$ is not optimal. Hence, there exists a transmitted data curve $B_{1}(t)\neq B(t)$, such that (i) $B_{1}(T)>B(T)$ or (ii) $B_{1}(T)=B(T)$ and $E_{1}(T)\leq E(T)$.
First, we show that when $B(t)<B_{1}(t)$ holds in an interval, then $B(t)$ is linear. Let $a=\sup\left\{ t:~(\forall x:~B_{1}(x)=B(x))|~0\leq x< t\right\}$. Thus, in $[0,a)$, we have $B(t)=B_{1}(t)$. Assume that $(b,c)\subseteq(a,T]$ is the first interval that we have $B(t)<B_{1}(t)$ and $B(t)$ is not linear in $(b,c)$. Thus, we can find a subinterval $(b+\epsilon,d)$ where the line passing through its endpoints is under $B_1(t)$ and satisfies the data causality condition. Now, let $E_{diff}=\min(E_{1}(t)-E(t))$ in $(b+\epsilon,d)$. Since $E(b+\epsilon)\leq E_1(b+\epsilon)$ ($E(b)< E_1(b)$ and assuming $p(b)< p_1(b)$ in $(b,b+\epsilon)$, we can find a sufficiently small interval $(b+\epsilon_{1},e)\subseteq(b+\epsilon,d)$ such that maximum energy difference between the line passing through the endpoints and $B(t)$ in $(b+\epsilon_{1},e)$ is less than or equal to $E_{diff}$. Therefore, this line does not violate the causality conditions which results in a contradiction. Therefore, if in an interval $B(t)<B_{1}(t)$ holds, then $B(t)$ must be linear. Thus, based on Lemma \ref{Jensen}, $B(t)$ uses less energy than $B_{1}(t)$ in this interval.

Now, if $B_{1}(t)\leq B(t)$ in $t\in (a,T]$, then $B_{1}(T)=B(T)$. Since $B(t)$ is convex and $B(t)$ and $B_{1}(t)$ are increasing in $t$, based on Lemma \ref{zafar}, we have $E(T)<E_{1}(T)$ which is a contradiction.

If for $t\in (a,T]$, always $B_{1}(t)\leq B(t)$ does not hold, then we define $t_{c}$ as:
    \begin{align}
    t_{c}=\sup \left \{t: B_{1}(t)=B(t)~\textrm{for}~ t<T  \right \}
    \end{align}
We obtain $E(t_{c})\leq E_{1}(t_{c})$. Because, in some intervals in $(a,t_c]$, $B(t)$ is either linear and uses less energy (based on Lemma \ref{Jensen}) or $B_{1}(t)\leq B(t)$ which again uses less energy (based on Lemma \ref{zafar}).

Now we have two cases, $t_{c}\neq T$ and $t_{c}= T$. For the first case, if $B_{1}(t)\leq B(t)$ for $t_{c}<t<T$, we have a contradiction and the proof is completed, because (i) $B_{1}(T)<B(T)$ or (ii) $B_{1}(T)=B(T)$ and $E(T)<E_{1}(T)$, which both are contradictions. If $B(t)<B_{1}(t)$ for $t_{c}<t<T$, then $B(t)$ must be linear in this interval. For this case, if $B(T)=B_{s}(T)$, then should $B_{1}(T)=B_{s}(T)$. Since the curve $B(t)$ uses less energy than $B_{1}(t)$ in $t_{c}<t<T$ and $E(t_{c})\leq E_{1}(t_{c})$, we get $E(T)< E_{1}(T)$ which is a contradiction. Having $E(T)=E_{s}(T)$ implies that $E_{s}(T)< E_{1}(T)$ which is a contradiction, too.
When $t_{c}= T$, we have an interval $(t_{b},T]$ in which $B_{1}(t)=B(t)$ and $t_{b}=\inf\left\{ t:~(\forall x:~B_{1}(x)=B(x))|~t< x\leq T\right\}$, then we can define $t_{m}$ as:
 \begin{align}
    t_{m}=\sup \left \{t: B_{1}(t)=B(t)~\textrm{for}~ t<t_{b}  \right \}
    \end{align}
Now, we can use the same argument with substituting of $t_{m}$ instead of $t_{c}$, and the proof is complete.
\end{proof}

\begin{theorem}The presented algorithm is optimal.
    \end{theorem}
    \begin{proof} The proof is directly obtained from Lemmas \ref{IV.3} and \ref{IV.4}.
    \end{proof}
 \section{Multi-Hoping: Throughput Maximization and Completion Time Minimization}
 In this section, we consider a multi-hop channel with one Tx, one Rx and many relays and we investigate the throughput maximization and completion time minimization problems in an offline model in a full-duplex mode. For simplicity we first assume that we have a two-hop communication channel which is illustrated in Fig. 4. Then we extend the results to $n$ relays in Corollaries 2 and 3.
 \subsection{Throughput Maximization}
 \begin{figure}
        \centering
        \includegraphics[width=3.3in]{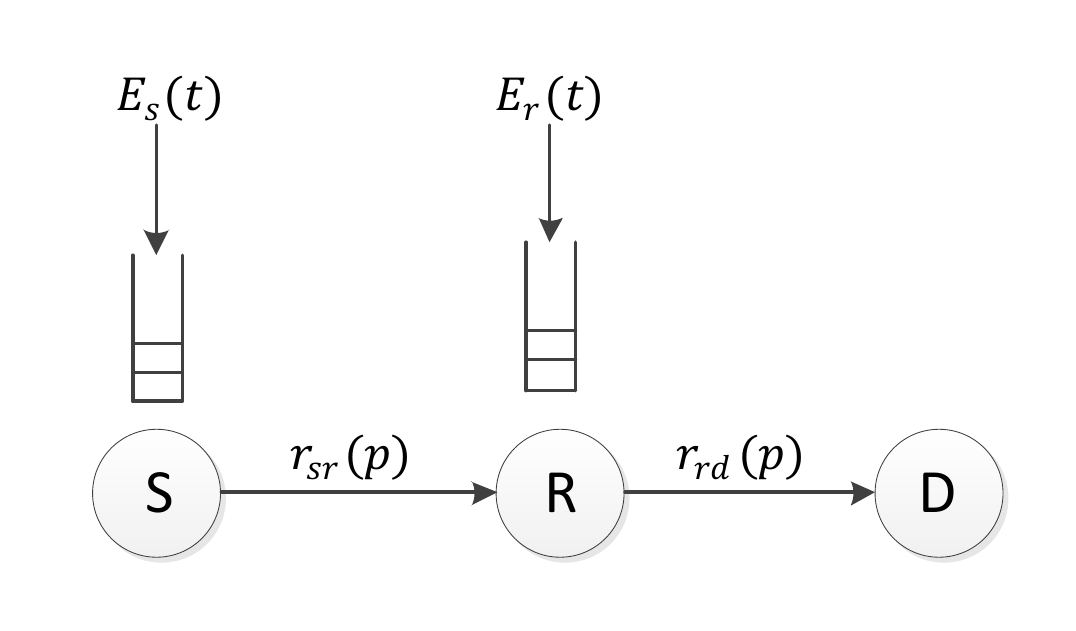}
        \caption{The topology of the network}
        \end{figure}

Following our assumption in Section II, in this model all harvested energy curves $E_{s}(t)$ and $E_{r}(t)$ and arrival data curve $B_{s}(t)$ are continuous. Similar to the single-user channel in Section II assume that the instantaneous transmission rates in both relay and Tx relate to the power of transmission through continuous functions $r_{sr}(p)$ and $r_{rd}(p)$, respectively. $B_{sr}(t)$ and $B_{rd}(t)$ are the amount of data which are transmitted from the Tx to the relay and from the relay to the Rx, respectively. $E_{sr}(t)$ and $E_{rd}(t)$ are the amount of energy that are utilized in the Tx and the relay to transmit data from the Tx to the relay, and the relay to Rx in $[0,t]$ respectively. $p_{sr}(t)$ and $p_{rd}(t)$ are the amount of power used in Tx, and the relay for data transmission. We assume that $B_{sr,s}^{*}(t)$ and $E_{sr,s}^{*}(t)$ are, respectively the optimal transmitted data and energy curves which are obtained from problem \eqref{kjhl}-\eqref{gh} (only when we consider the causality conditions in Tx); and $B_{rd,B_{sr}}^{*}(t)$, $E_{rd,B_{sr}}^{*}(t)$ and $p_{rd,B_{sr}}^{*}(t)$ are respectively the optimal transmitted data curve, optimal transmitted energy curve and optimal transmitted power curve which are obtained from problem \eqref{kjhl}-\eqref{gh} when we substitute $B_{s}(t)=B_{sr}(t)$ and $E_{s}(t)=E_{r}(t)$ (we consider the causality conditions in Relay). Now, we can formulate our problem as follows:
\vspace*{-.2cm}
\begin{eqnarray}
D^{(MH)}(T)&=&\!\!\!\!\!\max_{p_{sr}(t),p_{rd}(t)} \int_{0}^{T}r_{rd}(p_{rd}(t))dt\label{42}\\
 s.t.~~\int_{0}^{t}p_{sr}(t^{'})&\leq &\!\!\!\!\! E_{s}(t),~0\leq t\leq T\label{43}\\
\int_{0}^{t}p_{rd}(t^{'})&\leq &\!\!\!\!\! E_{r}(t),~0\leq t\leq T\label{44}\\
 \int_{0}^{t}r_{sr}(p_{sr}(t^{'}))dt^{'}&\leq &\!\!\!\!\! B_{s}(t),~0\leq t\leq T\label{45}\\
   \int_{0}^{t}r_{rd}(p_{rd}(t^{'}))dt^{'}&\leq &\!\!\!\!\!\!\!\int_{0}^{t}r_{sr}(p_{sr}(t^{'}))dt^{'},~0\leq t\leq T\label{46}.
\end{eqnarray}
 \eqref{43} and \eqref{44} are the energy causality conditions in Tx and the relay. \eqref{45} and \eqref{46} are the data causality conditions in Tx and the relay.  Also we assume that $B_{sr}^{*}(t)$ and $E_{sr}^{*}(t)$ are the optimal transmitted data and energy curves in Tx, $B_{rd}^{*}(t)$ and $E_{rd}^{*}(t)$ are the optimal transmitted data and energy curves in relay for problem \eqref{42}-\eqref{46}.
 
 In the following theorem, we show that the optimal solution of the two-hop transmission problem in \eqref{42}-\eqref{46} is derived by first solving a point-to-point throughput maximization problem at the source, and next solving a point-to-point throughput maximization problem at the relay (after applying the first solution as the input of the second problem).
 
 \begin{theorem}\label{V.1.}
In the optimal policy we have $B_{sr}^{*}(t)=B_{sr,s}^{*}(t)$ and $B_{rd}^{*}(t)=B_{rd,B_{sr,s}^{*}}^{*}(t)$.
 \end{theorem}

\begin{proof}
It is enough to prove that for any feasible $B_{sr}(t)$, we have: $B_{rd,B_{sr}}^{*}(t)\leq B_{rd,B_{sr,s}^{*}}^{*}(t)$. To prove, we use contradiction as well as the technique mentioned in Remark~\ref{rem:tech}. Thus, assume that $(a,b)$ is the first interval in which $ B_{rd,B_{sr,s}^{*}}^{*}(t)<B_{rd,B_{sr}}^{*}(t)$, and so $B_{rd,B_{sr,s}^{*}}^{*}(a)=B_{rd,B_{sr}}^{*}(a)$. Therefore, there exists an interval $(a,a+\epsilon)$ in which $p_{rd,B_{sr,s}^{*}}^{*}(t)<p_{rd,B_{sr}}^{*}(t)$.

 Now, we use \eqref{31} for an arbitrary $t_{0}\in(a,a+\epsilon)$ as follows:
\begin{align}p_{rd,B_{sr}}^{*}(t_{0})=\min\{\inf_{t_{0}<x\leq T}r^{-1}_{rd}(\frac{B_{sr}(x)-B_{rd,B_{sr}}^{*}(t_{0})}{x-t_{0}}),\nonumber\\\inf_{t_{0}<x\leq T}\frac{E_{r}(x)-E_{rd,B_{sr}}^{*}(t_{0})}{x-t_{0}}\}~~~~~~~\nonumber\\
 p_{rd,B_{sr,s}^{*}}^{*}(t_{0})=\min\{\inf_{t_{0}<x\leq T}r^{-1}_{rd}(\frac{B_{sr,s}^{*}(x)-B_{rd,B_{sr,s}^{*}}^{*}(t_{0})}{x-t_{0}}),\nonumber\\\inf_{t_{0}<x\leq T}\frac{E_{r}(x)-E_{rd,B_{sr,s}^{*}}^{*}(t_{0})}{x-t_{0}}\}.~~~~~~~
 \end{align}
 Also, we have,
 \begin{align}\label{ener}
 E_{rd,B_{sr,s}^{*}}^{*}(t_{0})=E_{rd,B_{sr,s}^{*}}^{*}(a)+\int_{a}^{t_{0}}p_{rd,B_{sr,s}^{*}}^{*}(t)dt\nonumber\\
 E_{rd,B_{sr,s}}^{*}(t_{0})=E_{rd,B_{sr,s}}^{*}(a)+\int_{a}^{t_{0}}p_{rd,B_{sr,s}}^{*}(t)dt
 \end{align}
Due to Lemma \ref{zafar} we have:
\begin{align}\label{ener2}
E_{rd,B_{sr,s}^{*}}^{*}(a)\leq E_{rd,B_{sr}}^{*}(a)
\end{align}
Based on \eqref{ener}, \eqref{ener2} and $p_{rd,B_{sr,s}^{*}}^{*}(t)<p_{rd,B_{sr}}^{*}(t),t\in(a,t_{0})$, we have $E_{rd,B_{sr,s}^{*}}^{*}(t_{0})<E_{rd,B_{sr}}^{*}(t_{0})$. This results in,
 \begin{align}\label{energys}
  \inf_{t_{0}<x\leq T}\frac{E_{r}(x)-E_{rd,B_{sr}}^{*}(t_{0})}{x-t_{0}}<\inf_{t_{0}<x\leq T}\frac{E_{r}(x)-E_{rd,B_{sr,s}^{*}}^{*}(t_{0})}{x-t_{0}}.
 \end{align}
Now, let $m=\argmin\limits_{t_{0}<x\leq T}(\frac{B_{sr}(x)-B_{rd,B_{sr}}^{*}(t_{0})}{x-t_{0}})$ and $x_{c}>t_{0}$ is a point where $B_{sr,s}^{*}(x_{c})<B_{sr}(x_{c})$. Then, based on Lemma \ref{III.7.} there exist $\epsilon_{1}>0$ and $\epsilon_{2}>0$ such that (i) in $(x_{c}-\epsilon_{1},x_{c}+\epsilon_{2})$: $B_{sr,s}^{*}(t)$ is linear and $B_{sr,s}^{*}(t)<B_{sr}(t)$, and (ii)
  $B_{sr,s}^{*}(x_{c}-\epsilon_{1})=B_{sr}(x_{c}-\epsilon_{1})$ and $B_{sr,s}^{*}(x_{c}+\epsilon_{2})=B_{sr}(x_{c}+\epsilon_{2})$.
Since in $(x_{c}-\epsilon_{1},x_{c}+\epsilon_{2})$ we have $B_{sr,s}^{*}(t)<B_{sr}(t)$, then $\frac{B_{sr,s}^{*}(t)-B_{rd,B_{sr}}^{*}(t_{0})}{t-t_{0}}<\frac{B_{sr}(t)-B_{rd,B_{sr}}^{*}(t_{0})}{t-t_{0}}$.
Now, we have two cases: (i) $\frac{B_{sr,s}^{*}(x_{c}-\epsilon_{1})-B_{rd,B_{sr}}^{*}(t_{0})}{x_{c}-\epsilon_{1}-t_{0}}\leq\frac{B_{sr,s}^{*}(x_{c}+\epsilon_{2})-B_{rd,B_{sr}}^{*}(t_{0})}{x_{c}+\epsilon_{2}-t_{0}}$, and (ii) $\frac{B_{sr,s}^{*}(x_{c}+\epsilon_{2})-B_{rd,B_{sr}}^{*}(t_{0})}{x_{c}+\epsilon_{2}-t_{0}}<\frac{B_{sr,s}^{*}(x_{c}-\epsilon_{1})-B_{rd,B_{sr}}^{*}(t_{0})}{x_{c}-\epsilon_{1}-t_{0}}$.
Because of linearity of $B_{sr,s}^{*}(t)$ in $(x_{c}-\epsilon_{1},x_{c}+\epsilon_{2})$, for the cases (i) and (ii), we have $\frac{B_{sr,s}^{*}(x_{c}-\epsilon_{1})-B_{rd,B_{sr}}^{*}(t_0)}{x_{c}-\epsilon_{1}-t_{0}}\leq\frac{B_{sr,s}^{*}(t)-B_{rd,B_{sr}}^{*}(t_{0})}{t-t_{0}}$, and $\frac{B_{sr,s}^{*}(x_{c}+\epsilon_{2})-B_{rd,B_{sr}}^{*}(t_{0})}{x_{c}+\epsilon_{2}-t_{0}}<\frac{B_{sr,s}^{*}(t)-B_{rd,B_{sr}}^{*}(t_{0})}{t-t_{0}}$ for $t\in[x_{c}-\epsilon_{1},x_{c}+\epsilon_{2}]$, respectively. Therefore, for the case (i) we have,
  \begin{align}
  \frac{B_{sr}(x_{c}-\epsilon_{1})-B_{rd,B_{sr}}^{*}(t_{0})}{x_{c}-\epsilon_{1}-t_{0}}=\frac{B_{sr,s}^{*}(x_{c}-\epsilon_{1})-B_{rd,B_{sr}}^{*}(t_{0})}{x_{c}-\epsilon_{1}-t_{0}}\nonumber\\\leq \frac{B_{sr,s}^{*}(t)-B_{rd,B_{sr}}^{*}(t_{0})}{t-t_{0}}<\frac{B_{sr}(t)-B_{rd,B_{sr}}^{*}(t_{0})}{t-t_{0}}~~~~~~
\end{align}
for $t\in(x_{c}-\epsilon_{1},x_{c}+\epsilon_{2})$. And, similarly, for the case (ii) we have,
 \begin{align}
   \frac{B_{sr}(x_{c}+\epsilon_{2})-B_{rd,B_{sr}}^{*}(t_{0})}{x_{c}+\epsilon_{2}-t_{0}}=\frac{B_{sr,s}^{*}(x_{c}+\epsilon_{2})-B_{rd,B_{sr}}^{*}(t_{0})}{x_{c}+\epsilon_{2}-t_{0}}\nonumber\\\leq \frac{B_{sr,s}^{*}(t)-B_{rd,B_{sr}}^{*}(t_{0})}{t-t_{0}}<\frac{B_{sr}(t)-B_{rd,B_{sr}}^{*}(t_{0})}{t-t_{0}}~~~~~~
 \end{align}
for $t\in(x_{c}-\epsilon_{1},x_{c}+\epsilon_{2})$. Hence, $m$ is not in the intervals that $B_{sr,s}^{*}(t)<B_{sr}(t)$. Now, the only candidates for $m$ are the instants where $B_{sr,s}^{*}(t)\geq B_{sr}(t)$. Hence, $B_{sr}(m)\leq B_{sr,s}^{*}(m)$ and
\begin{align}\label{datas}
 \inf_{t_{0}<x\leq T}r^{-1}_{rd}(\frac{B_{sr}(x)-B_{rd,B_{sr}}^{*}(t_{0})}{x-t_{0}})<\nonumber\\\inf_{t_{0}<x\leq T}r^{-1}_{rd}(\frac{B_{sr,s}^{*}(x)-B_{rd,B_{sr,s}^{*}}^{*}(t_{0})}{x-t_{0}}).
\end{align}
 From \eqref{energys} and \eqref{datas}, we have $p_{rd,B_{sr}}^{*}(t_{0})<p_{rd,B_{sr,s}^{*}}^{*}(t_{0})$ which is a contradiction. Hence, proof is complete.
\end{proof}
\begin{coro}
Theorem \ref{V.1.} can be extended to $n$ relays: Tx transmits maximum amount of data by proposed algorithm in Section IV, the first relay sends maximum amount of data to the second relay by the same algorithm and this procedure repeats until the Rx.
 \end{coro}

\textbf{An example:}
 Assume that harvested energy curves in Tx and the relay nodes are $E_{s}(t)=e^{t}-1$, $E_{r}(t)=2e^{t}-2$, respectively and $r_{sr}(p)=r_{rd}(p)=\frac{1}{2}\log(1+p)$ in which logarithm is in base 2. We want to maximize the throughput from Tx to the destination. Using energy causality and convexity of $E_{s}(t)$ based on Section IV instantaneous arrival data at the relay is maximized in every $t\in[0,1]$ if $E_{sr}(t)=E_{s}(t)$. Thus, the optimal arrival data at the relay is $B_{sr,s}^{*}(t)=\int_{0}^{t}\frac{1}{2}\log(1+\frac{d}{dt^{'}}E_{s}(t^{'}))dt^{'}$ which is a continuous curve. Now, the problem reduces to a single-user throughput maximization problem in the relay node with harvested energy curve $E_{r}(t)$, and arrival data curve $B_{sr,s}^{*}(t)$.
\subsection{Completion Time Minimization}
In this subsection we investigate a completion time minimization problem to transmit $B_{0}$ amount of data to Rx in a multi-hop channel. We remark that the results of this section can be easily reduced to the single-user scenario. We can formulate the problem as follows:
 \begin{align}T_{off}= \min~ T~~~~~~~~~~~~~~~~~~~~~~\\
 s.t.~ \int^{T}_{0} r_{rd}(p_{rd}(t))dt=B_{0},~
 \eqref{43}-\eqref{46}.
 \end{align}
 \begin{lemma}\label{V.2}
 $D^{(MH)}(t)$ in \eqref{42} is nondecreasing. Also if $\lim_{p\to \infty}\frac{r_{rd}(p)}{p}=0$, then $D^{(MH)}(t)$ is continuous.
 \end{lemma}
 \begin{proof}
  The proof of the first part is straightforward and is omitted for brevity. For the second part, for $t\in(t_{0},t_{0}+\epsilon]$ with any $\epsilon \geq0 $ we have,
  \begin{align}
  D^{(MH)}(t_{0})\leq D^{(MH)}(t)\leq D^{(MH)}(t_{0})+(t-t_{0}) r(\frac{A(t)}{t-t_{0}}),
  \end{align}
  where $A(t)=E_{s}(t)-E^{*}(t_0)+A_{0}$ and $\lim_{t\to t_{0}}A(t)=E_{s}(t_{0})-E^{*}(t_{0})+A_{0}=A_{1}$, and $0<A_{1}$ is a finite number. Based on $\lim_{p \to \infty}\frac{r(p)}{p}=0$ and assuming $p=\frac{A(t)}{t-t_{0}}$ we have,
  \begin{align}
  \lim_{t\to t_{0}^{+}} (D^{(MH)}(t_{0})+ (t-t_{0}) r(\dfrac{A(t)}{t-t_{0}}))\nonumber\\=D^{(MH)}(t_{0})+ A_{1}\lim_{p\to \infty}\frac{r(p)}{p}=D^{(MH)}(t_{0}).
  \end{align}
  From above, it is concluded $\lim_{t\to t_{0}^{+}}D^{(MH)}(t)=D^{(MH)}(t_{0})$. We can similarly prove that $\lim_{t\to t_{0}^{-}}D^{(MH)}(t)=D^{(MH)}(t_{0})$. Thus $D^{(MH)}(t)$ is continuous.
  \end{proof}
 \begin{theorem}\label{V.3}
 Assume that $C=\left \{t:~~D^{(MH)}(t)=B_{0}  \right \}$ and $\lim_{p\to \infty}\frac{r_{rd}(p)}{p}=0$. If $C\neq \emptyset$ and $T_{min}=\min~C$, then $T_{off}=T_{min}$, and optimal offline algorithm is given by the proposed algorithm in Section IV for given deadline $T_{off}$. If $C=\emptyset$, there does not exist any policy to transmit the amount of $B_{0}$ data.
 \end{theorem}
 \begin{proof}
Obviously, if $C\neq \emptyset$, exists a method to transmit amount of $B_{0}$ data in $T_{min}$. Based on Lemma \ref{V.2}, if $T_{off}< T_{min}$ holds, we have $D^{(MH)}(T_{off})<B_{0}$. Thus $T_{off}= T_{min}$. If we assume that $C=\emptyset$, we can conclude the amount of $B_{0}$ data cannot be transmitted by any time. Because, if there exists a time $T_{c}$ such that we can transmit $B_{0}$ amount of data until $T_{c}$, we get $B_{0}\leq D^{(MH)}(T_{c})$. Thus, Lemma \ref{V.2} concludes $C\neq \emptyset$.
 \end{proof}
 \begin{coro} We can extend Theorem \ref{V.3} to $n$ relays with defining $D^{(MH)}(t)$ as maximum amount of data curve in Rx in Theorem \ref{V.3}.
 \end{coro}
 \section{An Online Algorithm}
 In this section, we want to propose an online algorithm for the optimization problem proposed in Section II. In our online algorithm, we do not have any information about the future of two curves $B_{s}(t), E_{s}(t)$, (even the distributions of two processes $B_{s}(t), E_{s}(t)$  are unknown). First, we prove that the proposed online algorithm uses all of the energy or sends all of the data in the data buffer, and the transmitted power curve is a nondecreasing function similar to optimal offline algorithm. Then, we derive a lower bound on the ratio of the amounts of transmitted data in the online algorithm to the optimal offline algorithm.

   We express the online algorithm in the following.
\begin{align}\label{49}
p_{on}(t)=\min
\left\{r^{-1} (\dfrac{B_{rem}(t)}{T-t+\epsilon}),\dfrac{E_{rem}(t)}{T-t+\epsilon}\right\},
\end{align}where $B_{rem}(t)=B_{s}(t)-B_{on}(t)$, $E_{rem}(t)=E_{s}(t)-E_
{on}(t)$ and $\epsilon$ is chosen to make the $p_{on}(t)$ a bounded curve. Note that $\epsilon$ is a sufficiently small real number.

According to above, in our algorithm, if in time $t$ the amount of energy is the limiting element, then $p_{on}(t)$ is determined such that all of the remaining energy in $t$ is utilized with a fixed power until time $T$ and if in time $t$ the amount of information is the limiting element, then $p_{on}(t)$ is determined such that all of the remaining bits in $t$ are transmitted with a fixed rate until time $T$. In the following, we obtain $p_{on}(t) $ in parameters $ B_{s}(t) $, $ E_{s}(t) $, $T$. We assume that $t_{1},t_{2},...,t_{n}$ are instants in which the $p_{on}(t)$ switches from $ r^{-1} (\frac{B_{rem}(t)}{T-t+\epsilon}) $ to $ \frac{E_{rem}(t)}{T-t+\epsilon} $ or vice versa. We assume that in interval $(t_{i-1},t_{i})$ we have $p_{on}(t)=\frac{E_{rem}(t)}{T-t+\epsilon}$. Hence, in $(t_{i},t_{i+1})$ we have $p_{on}(t)=r^{-1}(\frac{B_{rem}(t)}{T-t+\epsilon})$. Thus,

\begin{align}\label{50}
p_{on}(t^{+}_{i})=r^{-1}\bigg(\dfrac{B_{s}(t^{+}_{i})-B_{on}(t_{i})}{T-t_{i}+\epsilon}\bigg),
\end{align}
and after some algebraic calculation we obtain
\begin{align}\label{51}
p_{on}(t)=r^{-1} \bigg( \int_{t^{+}_{i}}^{t}\dfrac{\frac{d}{dt^{'}}B_{s}(t^{'})}{T-t^{'}+\epsilon}dt^{'}+r(p_{on}(t_{i}^{+})) \bigg)
\end{align}
for $t_{i}<t<t_{i+1}$. Also, we can easily show that if in $t\in(t_{i-1},t_{i})$, $p_{on}(t)=r^{-1} (\frac{B_{rem}(t)}{T-t+\epsilon})$ holds, then in $t\in(t_{i},t_{i+1})$ we have $p_{on}(t)=\frac{E_{rem}(t)}{T-t+\epsilon}$ and so,
\begin{align}
p(t^{+}_{i})=\frac{E_{s}(t^{+}_{i})-E_{on}(t_{i})}{T-t_{i}+\epsilon}
\end{align}
\begin{align}\label{53}
p_{on}(t)=  \int_{t^{+}_{i}}^{t}\frac{\frac{d}{dt^{'}}E_{s}(t^{'})}{T-t^{'}+\epsilon}dt^{'}+p_{on}(t_{i}^{+})
\end{align}
for $t_{i}<t<t_{i+1}$.
\begin{lemma}\label{VI.1.}$p_{on}(t)$ is a nondecreasing function.
\end{lemma}
\begin{proof}\eqref{51} and \eqref{53} conclude that $p_{on}(t)$ is nondecreasing in all intervals $(t_{i},t_{i+1})$ for $0\leq i\leq n$ with $t_{0}=0, t_{n+1}=T$. Thus, we must only prove that $p(t_{i}^{-})\leq p(t_{i}^{+})$ for $0<i<n+1$. If in $(t_{i-1},t_{i})$, $p_{on}(t)=\dfrac{E_{rem}(t)}{T-t+\epsilon}$ holds, then in $(t_{i},t_{i+1})$ we have $p_{on}(t)=r^{-1}(\frac{B_{rem}(t)}{T-t+\epsilon})$. Therefore, \eqref{49} concludes that $\frac{E_{rem}(t^{-}_{i})}{T-t_{i}+\epsilon}\leq r^{-1}(\frac{B_{rem}(t_{i}^{-})}{T-t_{i}+\epsilon})$. Since $B_{on}(t^{-}_{i})=B_{on}(t^{+}_{i})$ and $B_{s}(t^{-}_{i})\leq B_{s}(t^{+}_{i})$, we have $\frac{B_{rem}(t^{-}_{i})}{T-t_{i}+\epsilon}\leq \frac{B_{rem}(t^{+}_{i})}{T-t_{i}+\epsilon}$ which results in $p_{on}(t_{i}^{-})\leq p_{on}(t_{i}^{+})$. If in $(t_{i-1},t_{i})$ we have $p_{on}(t)=r^{-1}(\frac{B_{rem}(t)}{T-t+\epsilon})$, similarly we can show that $p_{on}(t_{i}^{-})\leq p_{on}(t_{i}^{+})$, which completes the proof.\end{proof}
\begin{lemma}\label{VI.2.}In our online algorithm either $\lim\limits_{\epsilon\to 0}E_{on}(T)=E_{s}(T)$ or $\lim\limits_{\epsilon\to 0}B_{on}(T)=B_{s}(T)$. Moreover, if for $t\in(t_{n},T)$ we have $p_{on}(t)=r^{-1} (\frac{B_{rem}(t)}{T-t+\epsilon})$, then $\lim\limits_{\epsilon\to 0}B_{on}(T)=B_{s}(T)$; otherwise, $\lim\limits_{\epsilon\to 0}E_{on}(T)=E_{s}(T)$.
\end{lemma}
\begin{proof} We assume that for $t\in(t_{n},T)$ we have $p_{on}(t)=r^{-1} (\frac{B_{rem}(t)}{T-t+\epsilon})$, the other condition can be proved similarly.
\begin{align}
B_{on}(T)=\int_{0}^{T} r(p_{on}(t))dt=\int_{0}^{t_{n}}r(p_{on}(t))dt\nonumber\\
+\int_{t_{n}}^{T}r(p_{on}(t))dt=B_{on}(t_{n})+\int_{t_{n}}^{T}r(p_{on}(t))dt.
\end{align}
From \eqref{51}:
\begin{align}
 r(p_{on}(t))=
  \int_{t^{+}_{n}}^{t}\frac{\frac{d}{dt^{'}}B_{s}(t^{'})}{T-t^{'}+\epsilon}dt^{'}+\frac{B_{s}(t^{+}_{n})-B_{on}(t_{n})}{T-t_{n}+\epsilon}
 \end{align}
 for $t_{n}<t\leq T$. So,
 \begin{align}
 B_{on}(T)=&\int_{t_{n}}^{T}\int_{t_{n}^{+}}^{t}\frac{\frac{d}{dt^{'}}B_{s}(t^{'})}{T-t^{'}+\epsilon}dt^{'}dt\nonumber\\&+\frac{B_{s}(t^{+}_{n})-B_{on}(t_{n})}{T-t_{n}+\epsilon}(T-t_{n})+B_{on}(t_{n})=\nonumber\\
 &\int_{t^{+}_{n}}^{T}\frac{(T-t^{'})\frac{d}{dt^{'}}B_{s}(t^{'})}{T-t^{'}+\epsilon}dt^{'}+\frac{(T-t_{n})}{T-t_{n}+\epsilon}B_{s}(t_{n}^{+})+\nonumber\\
 &\frac{\epsilon}{T-t_{n}+\epsilon}B_{on}(t_{n}).~~~~~~~~~~~
 \end{align}
 Now, it is enough to prove that $\lim\limits_{\epsilon\to 0}(B_{s}(T)-B_{on}(T))=0$.
 \begin{align}
 B_{s}(T)-B_{on}(T)=B_{s}(T)-B_{s}(t_{n}^{+})+B_{s}(t_{n}^{+})-B_{on}(T)=~~~~~~\nonumber\\
 \int_{t^{+}_{n}}^{T}\bigg(\frac{d}{dt^{'}}B_{s}(t^{'})-\frac{(T-t^{'})\frac{d}{dt^{'}}B_{s}(t^{'})}{T-t^{'}+\epsilon}\bigg)dt^{'}+\dfrac{\epsilon}{T-t_{n}+\epsilon}B_{s}(t_{n}^{+})~~~~~~\nonumber\\
  -\frac{\epsilon}{T-t_{n}+\epsilon}B_{on}(t_{n})=\int_{t^{+}_{n}}^{T}\dfrac{\epsilon }{T-t^{'}+\epsilon}\frac{d}{dt^{'}}B_{s}(t^{'})dt^{'}+~~~~~~~~~~~\nonumber\\
  \frac{\epsilon}{T-t_{n}+\epsilon}B_{s}(t_{n}^{+})-
    \frac{\epsilon}{T-t_{n}+\epsilon}B_{on}(t_{n}).~~~~~~~~~~~~~~~~~~~
 \end{align}
 It is clear that $\lim\limits_{\epsilon\to 0}\dfrac{\epsilon}{T-t_{n}+\epsilon}B_{s}(t_{n}^{+})=0$ and $\lim\limits_{\epsilon\to 0}\frac{\epsilon}{T-t_{n}+\epsilon}B_{on}(t_{n})=0$, because, $B_{s}(t)$ is bounded and as a result $B_{on}(t)$ is bounded, too. Now, assume that $t_{j_{1}},~t_{j_{2}},.., t_{j_{m}}$ are all of instants in interval $(t_{n},T)$ such that $B_{s}(t_{j_{i}}^{-})\neq B_{s}(t_{j_{i}}^{+})$ for $1\leq i\leq m$. Thus,
 \begin{align}
 \int_{t^{+}_{n}}^{T}\frac{\epsilon }{T-t^{'}+\epsilon}\frac{d}{dt^{'}}B_{s}(t^{'})dt^{'}=\int_{t^{+}_{n}}^{t_{j_{1}}^{-}}\frac{\epsilon }{T-t^{'}+\epsilon}\frac{d}{dt^{'}}B_{s}(t^{'})dt^{'}+\nonumber\\
 \frac{\epsilon(B_{s}(t_{j_{1}}^{+})-B_{s}(t_{j_{1}}^{-}))}{T-t_{j_{1}}+\epsilon}+\int_{t_{j_{1}}^{+}}^{t_{j_{2}}^{-}}\frac{\epsilon }{T-t^{'}+\epsilon}\frac{d}{dt^{'}}B_{s}(t^{'})dt^{'}+\nonumber\\
  \frac{\epsilon(B_{s}(t_{j_{2}}^{+})-B_{s}(t_{j_{2}}^{-}))}{T-t_{j_{2}}+\epsilon}+...+\int_{t_{j_{m}}^{+}}^{T}\frac{\epsilon }{T-t^{'}+\epsilon}\frac{d}{dt^{'}}B_{s}(t^{'})dt^{'}.
 \end{align}
 It is clear that $\lim\limits_{\epsilon\to 0}\frac{\epsilon(B_{s}(t_{j_{i}}^{+})-B_{s}(t_{j_{i}}^{-}))}{T-t_{j_{i}}+\epsilon}=0$, for $0<i<m$.
 From assumptions in Section II for $B_{s}(t)$, we conclude that $\frac{d}{dt^{'}}B_{s}(t)$ is bounded in intervals $(t_{j_{i}},t_{j_{i+1}})$ for $0\leq i\leq m$ with $t_{j_{0}}=t_{n}$ and $t_{j_{m+1}}=T$. Thus,
 \begin{align}
 0\leq \int_{t_{j_{i}}^{+}}^{t_{j_{i+1}}^{-}}\frac{\epsilon }{T-t^{'}+\epsilon}\frac{d}{dt^{'}}B_{s}(t^{'})dt^{'}\leq \int_{t_{j_{i}}^{+}}^{t_{j_{i+1}}^{-}}\frac{\epsilon }{T-t^{'}+\epsilon}M dt^{'}\
 \end{align}
for $0\leq i\leq m$, where $\frac{d}{dt}B_{s}(t)\lvert_{t=t^{+}}\leq M$ and $\frac{d}{dt}B_{s}(t)\lvert_{t=t^{-}}\leq M$ for $0\leq t\leq T$ in which $\frac{d}{dt}B_{s}(t)\lvert_{t=t^{-}},~\frac{d}{dt}B_{s}(t)\lvert_{t=t^{+}}$ mean left and right derivatives of $B_{s}(t)$ in $t$.

  Also it can be easily shown that, $\lim\limits_{\epsilon\to 0}\int_{t_{j_{i}}^{+}}^{t_{j_{i+1}}^{-}}\frac{\epsilon }{T-t^{'}+\epsilon}M dt^{'}=0$ for $0\leq i\leq m$. Hence the proof is complete. \end{proof}
  \begin{theorem}\label{VI.3}
Assume $l$ and $k$ are two real numbers such that $t_{i}< \frac{T}{l}< t_{i+1}$ and $E_{s}(\frac{T}{l})=\frac{E_{s}(T)}{k}$. If $p_{on}(\frac{T}{l})=\frac{E_{rem}(\frac{T}{l})}{T-\frac{T}{l}+\epsilon}$, then $\frac{1}{k}(1-\frac{1}{l})\leq \frac{B_{on}(T)}{B_{off}(T)}$ otherwise $\frac{B_{s}(\frac{T}{l})}{B_{s}(T)}\leq \frac{B_{on}(T)}{B_{off}(T)}$.
  \end{theorem}
  \begin{proof} Note that if $p_{on}(\frac{T}{l})=\frac{E_{rem}(\frac{T}{l})}{T-\frac{T}{l}+\epsilon}$, then  \begin{align}\label{59} \frac{E_{s}(\frac{T}{l})}{T+\epsilon}\leq p_{on}(\frac{T}{l}),\end{align} because if \eqref{59} does not hold then should
  \begin{align}\label{60}
  \frac{E_{s}(\frac{T}{l})-E_{on}(\frac{T}{l})}{T-\frac{T}{l}+\epsilon}<\frac{E_{s}(\frac{T}{l})}{T+\epsilon}.
  \end{align}
   Also, due to the fact that $p_{on}(t)$ is a nondecreasing function (based on Lemma \ref{VI.1.}), we have,
   \begin{align}\label{100} \frac{E_{on}(\frac{T}{l})}{\frac{T}{l}}\leq\frac{E_{s}(\frac{T}{l})-E_{on}(\frac{T}{l})}{T-\frac{T}{l}+\epsilon}.
   \end{align}
    After some algebraic calculation, it can be shown that \eqref{60} and \eqref{100} result in a contradiction and so \eqref{59} holds. Also, we have,
    \begin{align}\label{63}
    B_{on}(T)=B_{on}(\frac{T}{l})+\int^{T}_{\frac{T}{l}}r(p_{on}(t))dt\geq\nonumber\\
     \int^{T}_{\frac{T}{l}}r(p_{on}(t))dt\geq (T-\frac{T}{l})r(p_{on}(\frac{T}{l})).
    \end{align}
    Also, we have from Lemma \ref{Jensen},
    \begin{align}\label{64}
    B_{off}(T)\leq T r(\frac {E_{s}(T)}{T}).
    \end{align}
    Thus, it is concluded from \eqref{59}, \eqref{63} and \eqref{64} and assuming that $\epsilon$ is small sufficiently,
    \begin{align}
    \frac{B_{on}(T)}{B_{off}(T)}\geq \frac{(T-\frac{T}{l})r(\frac{E_{s}(\frac{T}{l})}{T+\epsilon})}{T r(\frac {E_{s}(T)}{T})}=\nonumber\\
    \frac{(1-\frac{1}{l})r(\frac{E_{s}(\frac{T}{l})}{T+\epsilon})}{r(\frac {E_{s}(T)}{T})}>\frac{1}{k}(1-\frac{1}{l}).
    \end{align}
    Now, since $p_{on}(t)$ is a nondecreasing function in $t$, if $p_{on}(\frac{T}{l})=r^{-1}\bigg(\frac{B_{s}(\frac{T}{l})-B_{on}(\frac{T}{l})}{T-\frac{T}{l}+\epsilon}\bigg)$ then we have,
    \begin{align}
    B_{on}(T)\geq B_{on}(\frac{T}{l})+\frac{B_{s}(\frac{T}{l})-B_{on}(\frac{T}{l})}{T-\frac{T}{l}+\epsilon}(T-\frac{T}{l})\nonumber\\
    =\frac{T-\frac{T}{l}}{T-\frac{T}{l}+\epsilon}B_{s}(\frac{T}{l})+\frac{\epsilon}{T-\frac{T}{l}+\epsilon}B_{on}(\frac{T}{l})\geq~~~~~\nonumber\\
     \frac{T-\frac{T}{l}}{T-\frac{T}{l}+\epsilon}B_{s}(\frac{T}{l}).~~~~~~~~~~~~~~~~~~~~
    \end{align}
    Also we know $B_{off}(T)\leq B_{s}(T)$ and since $\epsilon$ is sufficiently small we obtain,
    \begin{align}
    \frac{B_{on}(T)}{B_{off}(T)}\geq \frac{B_{s}(\frac{T}{l})}{B_{s}(T)},
    \end{align}
    and the proof of the theorem is complete.\end{proof}
\begin{coro}
The following results are concluded from Theorem \ref{VI.3}:

(i) If in $(t_{n},T)$, $p_{on}(t)=r^{-1}(\frac{B_{rem}(t)}{T-t+\epsilon})$, then $\frac{B_{on}(T)}{B_{off}(T)}\approx 1$, which means that the online algorithm transmits all of data bits of those in offline algorithm.

(ii) If $E_{s}(\frac{T}{2})=E_{s}(T)$ and $p_{on}(\frac{T}{2})=r^{-1}(\frac{E_{rem}(\frac{T}{2})}{T-\frac{T}{2}+\epsilon})$, then $\frac{B_{on}(T)}{B_{off}(T)}\geq\dfrac{1}{2}$, which means that the online algorithm transmits at least half of data bits of those in offline algorithm.
\end{coro}
\begin{proof}
For (i), in Theorem \ref{VI.3}, $l$ can be chosen sufficiently close to $1$. (ii) can be derived by substitution.
\end{proof}
Although, there are many examples that this bound is good for them but the authors believe that the above lower bound is not tight enough for all the arbitrary two curves $E_{s}(t),~B_{s}(t)$, and the algorithm is more efficient than the bound in these examples. Another advantage of this online algorithm is that it does not require any information about the distributions of the two processes $B_{s}(t)$ and $E_{s}(t)$.

 \section{Numerical Results}
 \begin{figure*}
  \centering
  \includegraphics[width=7in]{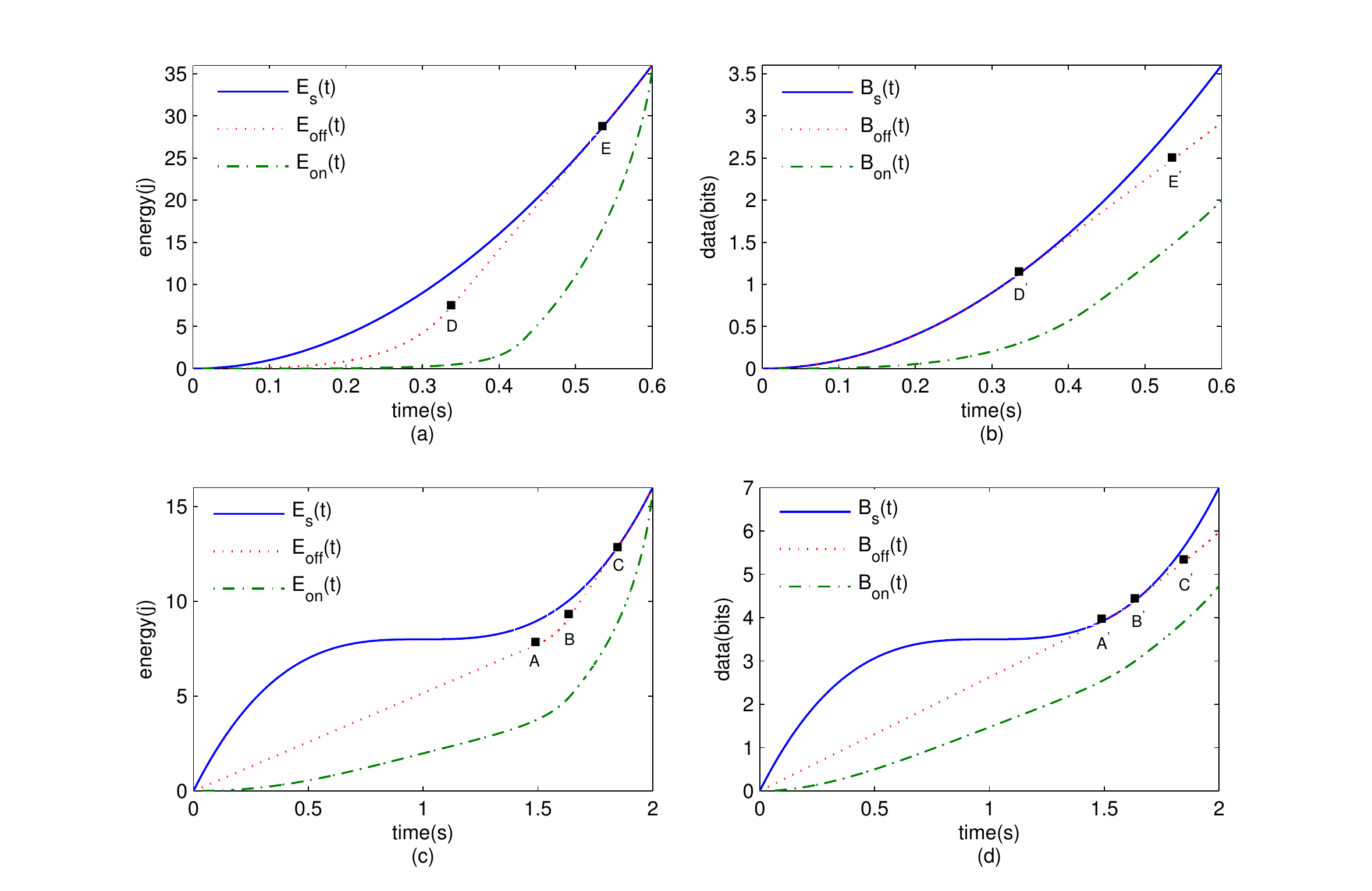}
  \caption{Online and Optimal offline algorithms without discretizing}
  \label{fig}
  \end{figure*}
  \begin{figure*}
    \centering
    \includegraphics[width=7in]{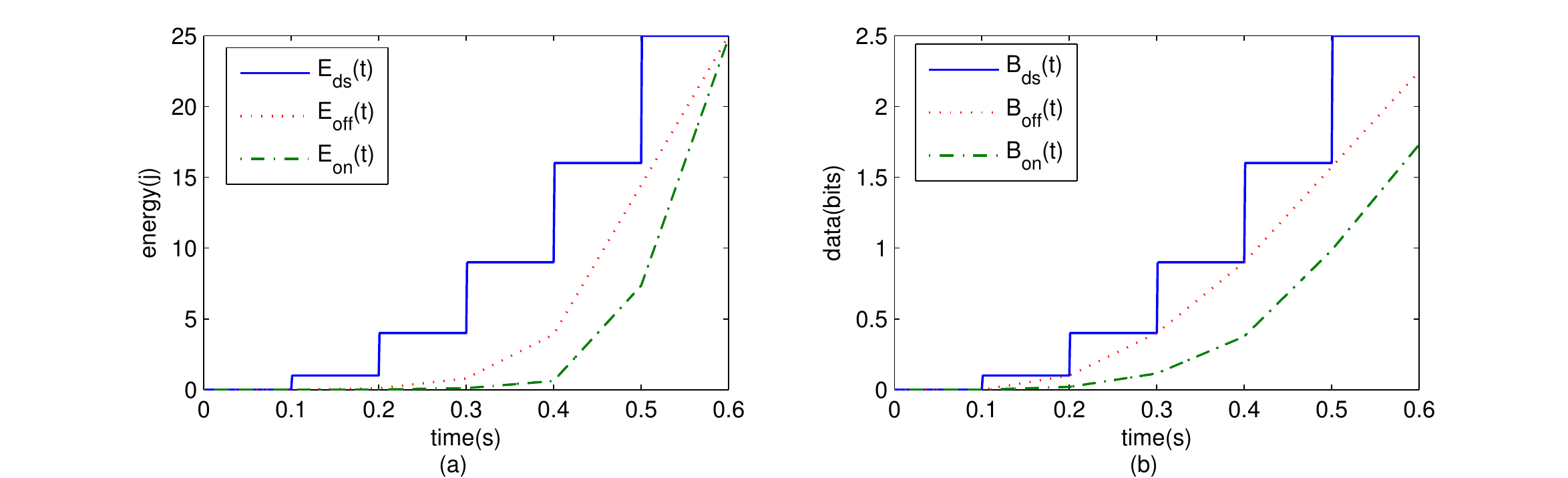}
    \caption{Online and Optimal offline algorithms with discretizing}
    \label{figd}
    \end{figure*}
In this section, we provide numerical examples to explain our results. Consider a band-limited additive white Gaussian noise channel with bandwidth $W=1$ Hz. Also, the actual channel gain divided by the noise power spectral density multiplied by the bandwidth is $1$. So we have, $r(p)=\log(1+p)$, where the logarithm is in base $2$. We consider two different $E_{s}(t),B_{s}(t)$ curve pairs. The first pair consists of two convex functions $E_{s}(t)=100 t^{2}$ J and $B_{s}(t)=10 t^{2}$ bits in Fig. \ref{fig} (a) and Fig. \ref{fig} (b), while the second pair consists of more general functions $E_{s}(t)=8 (t-1)^{3}+8$ J and $B_{s}(t)=3.5 (t-1)^{3}+3.5$ bits in Fig. \ref{fig} (c) and Fig. \ref{fig} (d). We remark that due to the nature of the harvested energy and arrival data, these functions must be non-decreasing. These figures show the harvested energy/arrival data curves and the transmitted energy/data curves based on our proposed offline and online algorithms versus the time. In Fig. \ref{fig} (a) and Fig. \ref{fig} (b), we assume that $T=0.6$ s while in Fig. \ref{fig} (c) and Fig. \ref{fig} (d), we assume that $T=2$ s.

As can be easily seen from Fig. \ref{fig} (a) and Fig. \ref{fig} (b), the $E_{off}(t)$ curve consists of three parts: the offline algorithm is in state C in $(0,D)$ (approximately $(0,.34)$), it is in state A in $(D,E)$ (approximately $(.34,0.54)$), and it is in state B in $(E,T)$ (approximately $(0.54,0.6)$). In $(0,D)$, $E_{off}(t)$ is nonlinear; $E_{off}(t)\neq E_{s}(t)$ and this means that $B(t)=B_{s}(t)$, according to Lemma \ref{sfp}; in $(E,T)$, $E_{off}(t)$ is nonlinear; $B_{off}(t)\neq B_{s}(t)$ and this means that $E(t)=E_{s}(t)$, according to Lemma \ref{sfp}. Moreover, we observe that $p(t)$ is a nondecreasing function because $\frac{d}{dt}E_{off}(t)\geq 0$ (Lemma \ref{sdddf}).
We remark that the optimal algorithm can transmit at most $2.9$ bits at the end of the interval ($T=0.6$) while it uses all the harvested energy, i.e., $E_{off}(0.6)=E_{s}(0.6)$. Thus, the system is "energy constrained" in this case.

As can be easily seen from Fig. \ref{fig} (c) and Fig. \ref{fig} (d), the $E_{off}(t)$ curve consists of four parts: the offline algorithm is in state A in $(0,A)$ (approximately $(0,.1.5)$), it is in state C in $(A,B)$ (approximately $(1.5,1.63)$), it is in state A in $(B,C)$ (approximately $(1.63,1.86)$) and it is in state B in $(C,T)$ (approximately $(1.86,2)$). In $(A,B)$, $E_{off}(t)$ is nonlinear; $E_{off}(t)\neq E_{s}(t)$ and this means that $B(t)=B_{s}(t)$, according to Lemma \ref{sfp}; in $(C,T)$, $E_{off}(t)$ is nonlinear; $B_{off}(t)\neq B_{s}(t)$ and this means that $E(t)=E_{s}(t)$, according to Lemma \ref{sfp}. Moreover, we observe that $p(t)$ is a nondecreasing function because $\frac{d}{dt}E_{off}(t)\geq 0$ (Lemma \ref{sdddf}).
We remark that the optimal algorithm can transmit at most $6$ bits at the end of the interval ($T=2$) while it uses all the harvested energy, i.e., $E_{off}(0.6)=E_{s}(0.6)$. Thus, the system is "energy constrained" in this case.

For the online algorithm we assume that $\epsilon=.001$. It can be easily seen that $E_{on}(t)$ and $B_{on}(t)$ are convex hence, $p_{on}(t)$ is nondecreasing (Lemma \ref{VI.1.}). Also, in Fig. \ref{fig} (a) $E_{on}(0.6)\approx E_{s}(0.6)$ and in Fig. \ref{fig} (c) $E_{on}(2)\approx E_{s}(2)$ (Lemma \ref{VI.2.}). From Fig. \ref{fig} (b) we have $\frac{B_{on}(T)}{B_{off}(T)}\approx\frac{2}{2.9}$ which means that approximately $69$ percent of data that transmitted by offline algorithm is transmitted by online algorithm and from Fig. \ref{fig} (d) we have $\frac{B_{on}(T)}{B_{off}(T)}\approx\frac{4.8}{6}$ which means that approximately $80$ percent of data that transmitted by offline algorithm is transmitted by online algorithm (using all harvested energy in both algorithms).

As mentioned in Section I, Fig. \ref{figd} shows the necessity of investigating continuous model instead of discretizing harvested energy and arrival data curves to achieve the optimal performance. In Fig. \ref{figd}, the $E_{ds}(t)$ and $B_{ds}(t)$ are the discretized version of the $E_{s}(t)$ and $B_{s}(t)$ in Fig. \ref{fig} (a) and Fig. \ref{fig} (b), respectively. It can be easily seen that the optimal offline algorithm with discretizing transmits $2.25$ bits (compared to 2.9 bits in continuous model) which reduces the efficiency.
\section{Discussion and Conclusion}
In this paper, we considered an EH system with continuous arrival data and continuous harvested energy curves; while, most of the research in this area considered a discrete model due to the mathematical tractability of the ensuing system optimization.
Our work can be compared to the ones in \cite{varan2014energy,yang2012optimal,vaze2014dynamic}.
In \cite{varan2014energy}, a model with continuous harvested energy curve is investigated, while it is assumed that the large amount of data exists to transmit (no arrival data process). Compared to our model, in \cite{varan2014energy} the causality condition of \eqref{gh} does not exist. Thus, the model in this paper is more general than \cite{varan2014energy}. In addition, \cite{varan2014energy} focuses on battery imperfection and processing gain which makes the results completely different from our results.
In \cite{yang2012optimal}, only the model with discrete $E_{s}(t)$ and $B_{s}(t)$ curves is investigated and its goal is to find the optimal policy that minimizes the completion time for transmitting a given amount of data among piecewise linear curves; while, in this paper we consider a model that includes both discrete and continuous models for $E_{s}(t)$ and $B_{s}(t)$, and find optimal policy among \emph{all} of curves assumed in Section II, therefore, the considered model of this paper is more general than \cite{yang2012optimal}. We compared our results thoroughly in Sections III and IV.
In \cite{vaze2014dynamic}, the optimal online algorithm for the discrete case with no data arrival is proposed which derives the transmitted power based on the available energy in the buffer (the sole constraint in this case). Our proposed online algorithm easily reduces to the mentioned algorithm by making the data available at the beginning and discretizing the harvested energy curve.

\bibliographystyle{./IEEEtran}
\bibliography{./IEEEabrv,./JSAC}

\begin{thebibliography}{10}
\providecommand{\url}[1]{#1}
\csname url@samestyle\endcsname
\providecommand{\newblock}{\relax}
\providecommand{\bibinfo}[2]{#2}
\providecommand{\BIBentrySTDinterwordspacing}{\spaceskip=0pt\relax}
\providecommand{\BIBentryALTinterwordstretchfactor}{4}
\providecommand{\BIBentryALTinterwordspacing}{\spaceskip=\fontdimen2\font plus
\BIBentryALTinterwordstretchfactor\fontdimen3\font minus
  \fontdimen4\font\relax}
\providecommand{\BIBforeignlanguage}[2]{{%
\expandafter\ifx\csname l@#1\endcsname\relax
\typeout{** WARNING: IEEEtran.bst: No hyphenation pattern has been}%
\typeout{** loaded for the language `#1'. Using the pattern for}%
\typeout{** the default language instead.}%
\else
\language=\csname l@#1\endcsname
\fi
#2}}
\providecommand{\BIBdecl}{\relax}
\BIBdecl

\bibitem{ozel2010information}
O.~Ozel and S.~Ulukus, ``Information-theoretic analysis of an energy harvesting
  communication system,'' in \emph{Personal, Indoor and Mobile Radio
  Communications Workshops (PIMRC Workshops), 2010 IEEE 21st International
  Symposium on}.\hskip 1em plus 0.5em minus 0.4em\relax IEEE, 2010, pp.
  330--335.

\bibitem{rajesh2011information}
R.~Rajesh, V.~Sharma, and P.~Viswanath, ``Information capacity of energy
  harvesting sensor nodes,'' in \emph{Information Theory Proceedings (ISIT),
  2011 IEEE International Symposium on}.\hskip 1em plus 0.5em minus 0.4em\relax
  IEEE, 2011, pp. 2363--2367.

\bibitem{dong2014approximate}
Y.~Dong and A.~Ozgur, ``Approximate capacity of energy harvesting communication
  with finite battery,'' in \emph{Information Theory (ISIT), 2014 IEEE
  International Symposium on}.\hskip 1em plus 0.5em minus 0.4em\relax IEEE,
  2014, pp. 801--805.

\bibitem{tutuncuoglu2014improved}
K.~Tutuncuoglu, O.~Ozel, A.~Yener, and S.~Ulukus, ``Improved capacity bounds
  for the binary energy harvesting channel,'' in \emph{Information Theory
  (ISIT), 2014 IEEE International Symposium on}.\hskip 1em plus 0.5em minus
  0.4em\relax IEEE, 2014, pp. 976--980.

\bibitem{ozel2011transmission}
O.~Ozel, K.~Tutuncuoglu, J.~Yang, S.~Ulukus, and A.~Yener, ``Transmission with
  energy harvesting nodes in fading wireless channels: Optimal policies,''
  \emph{Selected Areas in Communications, IEEE Journal on}, vol.~29, no.~8, pp.
  1732--1743, 2011.

\bibitem{devillers2012general}
B.~Devillers and D.~Gunduz, ``A general framework for the optimization of
  energy harvesting communication systems with battery imperfections,''
  \emph{Communications and Networks, Journal of}, vol.~14, no.~2, pp. 130--139,
  2012.

\bibitem{tutuncuoglu2012optimum}
K.~Tutuncuoglu and A.~Yener, ``Optimum transmission policies for battery
  limited energy harvesting nodes,'' \emph{Wireless Communications, IEEE
  Transactions on}, vol.~11, no.~3, pp. 1180--1189, 2012.

\bibitem{arafa2014single}
A.~Arafa and S.~Ulukus, ``Single-user and multiple access channels with energy
  harvesting transmitters and receivers,'' in \emph{Signal and Information
  Processing (GlobalSIP), 2014 IEEE Global Conference on}.\hskip 1em plus 0.5em
  minus 0.4em\relax IEEE, 2014, pp. 213--217.

\bibitem{yang2012optimal}
J.~Yang and S.~Ulukus, ``Optimal packet scheduling in an energy harvesting
  communication system,'' \emph{Communications, IEEE Transactions on}, vol.~60,
  no.~1, pp. 220--230, 2012.

\bibitem{yang2012broadcasting}
J.~Yang, O.~Ozel, and S.~Ulukus, ``Broadcasting with an energy harvesting
  rechargeable transmitter,'' \emph{Wireless Communications, IEEE Transactions
  on}, vol.~11, no.~2, pp. 571--583, 2012.

\bibitem{nagda2014optimal}
R.~Nagda, S.~Satpathi, and R.~Vaze, ``Optimal offline and competitive online
  strategies for transmitter-receiver energy harvesting,'' \emph{arXiv preprint
  arXiv:1410.1292}, 2014.

\bibitem{gurakan2013energy}
B.~Gurakan, O.~Ozel, J.~Yang, and S.~Ulukus, ``Energy cooperation in energy
  harvesting communications,'' \emph{Communications, IEEE Transactions on},
  vol.~61, no.~12, pp. 4884--4898, 2013.

\bibitem{luo2013optimal}
Y.~Luo, J.~Zhang, and K.~B. Letaief, ``Optimal scheduling and power allocation
  for two-hop energy harvesting communication systems,'' \emph{Wireless
  Communications, IEEE Transactions on}, vol.~12, no.~9, pp. 4729--4741, 2013.

\bibitem{gunduz2011two}
D.~Gunduz and B.~Devillers, ``Two-hop communication with energy harvesting,''
  in \emph{Computational Advances in Multi-Sensor Adaptive Processing (CAMSAP),
  2011 4th IEEE International Workshop on}.\hskip 1em plus 0.5em minus
  0.4em\relax IEEE, 2011, pp. 201--204.

\bibitem{orhan2012optimal}
O.~Orhan and E.~Erkip, ``Optimal transmission policies for energy harvesting
  two-hop networks,'' in \emph{Information Sciences and Systems (CISS), 2012
  46th Annual Conference on}.\hskip 1em plus 0.5em minus 0.4em\relax IEEE,
  2012, pp. 1--6.

\bibitem{gurakan2014energy}
B.~Gurakan and S.~Ulukus, ``Energy harvesting diamond channel with energy
  cooperation,'' in \emph{Information Theory (ISIT), 2014 IEEE International
  Symposium on}.\hskip 1em plus 0.5em minus 0.4em\relax IEEE, 2014, pp.
  986--990.

\bibitem{huang2013throughput}
C.~Huang, R.~Zhang, and S.~Cui, ``Throughput maximization for the {Gaussian}
  relay channel with energy harvesting constraints,'' \emph{Selected Areas in
  Communications, IEEE Journal on}, vol.~31, no.~8, pp. 1469--1479, 2013.

\bibitem{feghhi2013optimal}
M.~M. Feghhi, A.~Abbasfar, and M.~Mirmohseni, ``Optimal power and rate
  allocation in the degraded {Gaussian} relay channel with energy harvesting
  nodes,'' in \emph{Communication and Information Theory (IWCIT), 2013 Iran
  Workshop on}.\hskip 1em plus 0.5em minus 0.4em\relax IEEE, 2013, pp. 1--6.

\bibitem{vaze2014dynamic}
R.~Vaze, R.~Garg, and N.~Pathak, ``Dynamic power allocation for maximizing
  throughput in energy-harvesting communication system,'' \emph{IEEE/ACM
  Transactions on Networking (TON)}, vol.~22, no.~5, pp. 1621--1630, 2014.

\bibitem{vaze2013competitive}
R.~Vaze, ``Competitive ratio analysis of online algorithms to minimize packet
  transmission time in energy harvesting communication system,'' in
  \emph{INFOCOM, 2013 Proceedings IEEE}.\hskip 1em plus 0.5em minus 0.4em\relax
  IEEE, 2013, pp. 115--1123.

\bibitem{xu2014throughput}
J.~Xu and R.~Zhang, ``Throughput optimal policies for energy harvesting
  wireless transmitters with non-ideal circuit power,'' \emph{Selected Areas in
  Communications, IEEE Journal on}, vol.~32, no.~2, pp. 322--332, 2014.

\bibitem{lei2009generic}
J.~Lei, R.~Yates, and L.~Greenstein, ``A generic model for optimizing
  single-hop transmission policy of replenishable sensors,'' \emph{Wireless
  Communications, IEEE Transactions on}, vol.~8, no.~2, pp. 547--551, 2009.

\bibitem{mao2012near}
Z.~Mao, C.~E. Koksal, and N.~B. Shroff, ``Near optimal power and rate control
  of multi-hop sensor networks with energy replenishment: Basic limitations
  with finite energy and data storage,'' \emph{Automatic Control, IEEE
  Transactions on}, vol.~57, no.~4, pp. 815--829, 2012.

\bibitem{wang2014power}
Z.~Wang, V.~Aggarwal, and X.~Wang, ``Power allocation for energy harvesting
  transmitter with causal information,'' \emph{Communications, IEEE
  Transactions on}, vol.~62, no.~11, pp. 4080--4093, 2014.

\bibitem{varan2014energy}
B.~Varan, K.~Tutuncuoglu, and A.~Yener, ``Energy harvesting communications with
  continuous energy arrivals,'' in \emph{Information Theory and Applications
  Workshop (ITA), 2014}.\hskip 1em plus 0.5em minus 0.4em\relax IEEE, 2014, pp.
  1--10.

\bibitem{ottman2002adaptive}
G.~K. Ottman, H.~F. Hofmann, A.~C. Bhatt, G.~Lesieutre \emph{et~al.},
  ``Adaptive piezoelectric energy harvesting circuit for wireless remote power
  supply,'' \emph{Power Electronics, IEEE Transactions on}, vol.~17, no.~5, pp.
  669--676, 2002.

\bibitem{palanki2004rateless}
R.~Palanki and J.~S. Yedidia, ``Rateless codes on noisy channels,'' in
  \emph{IEEE International Symposium on Information Theory}.\hskip 1em plus
  0.5em minus 0.4em\relax Citeseer, 2004, pp. 37--37.

\bibitem{le2001network}
J.-Y. Le~Boudec and P.~Thiran, \emph{Network calculus: a theory of
  deterministic queuing systems for the internet}.\hskip 1em plus 0.5em minus
  0.4em\relax Springer Science \& Business Media, 2001, vol. 2050.

\bibitem{jeffrey2007table}
A.~Jeffrey and D.~Zwillinger, \emph{Table of integrals, series, and
  products}.\hskip 1em plus 0.5em minus 0.4em\relax Academic Press, 2007.

\bibitem{zafer2009calculus}
M.~A. Zafer and E.~Modiano, ``A calculus approach to energy-efficient data
  transmission with quality-of-service constraints,'' \emph{IEEE/ACM
  Transactions on Networking (TON)}, vol.~17, no.~3, pp. 898--911, 2009.

\end{thebibliography}

\end{document}